%% file: main.tex
\documentclass[reqno,11pt]{amsart}
\usepackage{graphicx} % Required for inserting images

\usepackage{hyperref}
\usepackage{graphicx}
\usepackage{slashed}
\usepackage{amscd}
\usepackage{tensor}
\usepackage{amsmath,amssymb,mathabx,bm}
\usepackage{esint}
\usepackage{dsfont}
\usepackage{enumitem}
\usepackage{accents}
\usepackage{color}
\usepackage{accents}
\usepackage[giveninits=true]{biblatex}
\usepackage{tikz}
\usetikzlibrary{decorations.pathmorphing}
\usepackage[off]{auto-pst-pdf} % schnelle Version, Bilder müssen dazu schon vorliegen
\usepackage{pst-grad} % For gradients
\usepackage{pst-plot} % For axes
\usepackage[mathscr]{eucal}
\numberwithin{equation}{section}
\allowdisplaybreaks[1]

\title{On the Conicality of Spacetimes}
\author[E. Ling]{Eric Ling$^{*}$}

    \author[C. F. Paganini]{Claudio F. Paganini$^{**}$}
    \email{eling@fsu.edu \\ claudio.paganini@mathematik.uni-regensburg.de}

    \address{$^{*}$ Department of Mathematics \\ Florida State University \\ Tallahassee \\ Florida}

    \address{$^{**}$ Fakultät für Mathematik \\ Universität Regensburg \\ D-93040 Regensburg \\ Germany}
\newtheorem{Def}{Definition}[section]
\newtheorem{Thm}[Def]{Theorem}

\newtheorem{Prp}[Def]{Proposition}
\newtheorem{Lemma}[Def]{Lemma}
\newtheorem{Remark}[Def]{Remark}
\newtheorem{Corollary}[Def]{Corollary}
\newtheorem{Example}[Def]{Example}

\newtheorem*{theorem*}{Main Theorem}

\newcounter{mnotecount}[subsection]

\newcommand{\g}{\gamma}

\begin{document}
\begin{abstract}

 The notion of conicality, recently introduced in~\cite{grothus2024characterizing}, captures the extent to which the joint future of a finite set in spacetime uniquely determines the generating subset via its light cone structure. In the same paper it was mentioned that conicality holds for Minkowski spacetime of dimension $1+N$ with $N\geq 2$ and it has been conjectured that this property holds more generally.

In this work, we show that neither homotopy with Minkowski space nor global hyperbolicity alone are sufficient for the spacetime to satisfy conicality. We then establish that causally simple, future cohesive spacetimes of dimension $1+N$ with $N\geq 2$ are conical. This class of spacetimes includes, in particular, TIP spacetimes, which can be understood as the timelike past of an observer. This is in line with the origin of causal modeling since the past of an observer is the natural domain for the description of experiments. 
\end{abstract}
\maketitle

\vspace{-0.5cm}

\section{Introduction}

The causal structure of a Lorentzian manifold describes the ordering of events in a spacetime. In recent work~\cite{grothus2024characterizing}, the notion of \emph{conicality} was introduced as an order-theoretic property capturing the extent to which the joint future of a finite set of points determines the generating subset via its light cone structure. This concept arises naturally at the interface between Lorentzian geometry and causal modeling, where one seeks to understand how abstract causal relations can be faithfully embedded into spacetime.

More precisely, conicality expresses a rigidity property of the causal future operation: the intersection of future sets of finitely many events encodes sufficient information to reconstruct the subset of generators. In~\cite{grothus2024characterizing}, it was shown that Minkowski spacetime of dimension $1+N$, with $N \geq 2$, satisfies conicality, while the property fails in $1+1$ dimensions. Based on these observations, it was conjectured that conicality persists for a significantly broader class of spacetimes satisfying suitable causal regularity assumptions.

The motivation for this conjecture is rooted in causal modeling. In frameworks such as those developed in~\cite{spirtes1993causation,pearl2009causality,schoolmaster2020graphical,laubach2021biologists,kleinberg2011review,raita2021bigdata,arti2020trend,allen2017quantum,barrett2019quantum}, causal relations are represented abstractly via directed structures encoding dependencies among variables. Embedding such structures into spacetimes~\cite{vilasini2022general,vilasini2024embedding,grothus2024characterizing} requires compatibility between the combinatorial notion of causality and the geometric one induced by light cones. 

The definition of conicality is provided in Definition \ref{def:conicality}, but we give a brief explanation here. If $\mathcal{T}$ is a poset, we define $J^+(p) = \{q \in \mathcal{T} \mid p \preceq q\}$. Let $L = \{p_1, \dotsc, p_n\} \subset \mathcal{T}$ be a finite nonempty set. The joint future of $L$ is defined as $JF(L) = J^+(p_1) \cap \dotsb \cap J^+(p_n)$. The $\text{span}(L)$ is defined as the union of all nonempty sets $L'$ such that $JF(L') = JF(L)$ and there does not exist a nonempty $L'' \subsetneq L'$ such that $JF(L'') = JF(L)$. Then we say $\mathcal{T}$ is \emph{conical} if for any two finite nonempty subsets $L_i,L_j \subset \mathcal{T}$, we have 
\[
\text{span}(L_i) = \text{span}(L_j) \quad \Longleftrightarrow \quad JF(L_i) = JF(L_j).
\]
The implication ``$\Longrightarrow$" always holds, while the implication ``$\Longleftarrow$" is the challenging part.

Conicality can be interpreted as a condition ensuring that joint future regions behave in a manner consistent with causal inference, preventing degeneracies in which different generating sets produce indistinguishable observable futures.

In this paper, we will prove the following theorem, relaxing the conjectured causality condition from globally hyperbolic in~\cite{grothus2024characterizing} to causally simple. 
\begin{theorem*}
 Let  $(\mathcal{M},g)$ be a causally simple and future cohesive spacetime  of dimension $1+N$ with $N\geq 2$. Then $(\mathcal{M},g)$ is conical according to Definition \ref{def:conicality}.
\end{theorem*}
In Example \ref{ex:counterexample}, we will demonstrate that global hyperbolicity alone is not sufficient for conicality to hold. Rather it is the combination of causal simplicity and \emph{future cohesiveness} (see Definition \ref{def:futureconnected}) which is sufficient for conicality. Briefly, future cohesiveness is the requirement that all joint futures in a spacetime are connected. This property is satisfied for any TIP, i.e. the timelike past of an inextendible observer in the spacetime (see Proposition \ref{prp: TIP connected}). Therefore, conicality holds for any experimental data an observer (i.e. an entity conducting experiments) will ever collect, as all real world data has to originate from somewhere within the observers past light cone. Hence, this result covers the class of spacetimes that are most relevant in the context where the conjecture originated.

The result therefore strengthens the conceptual link between Lorentzian geometry and causal modeling introduced in~\cite{grothus2024characterizing} by identifying precise mathematical conditions under which spacetimes can faithfully encode abstract causal structure.

The proof of our main result (Theorem \ref{thm: main}) is entirely based on methods from Lorentzian causality theory, which is a rich field of on-going mathematical research, see, e.g.,~\cite{minguzzi2019lorentzian} and references therein. Causal simplicity ensures that the causal relation is closed and that null geodesics in the boundary of the causal future of a point intersect the point itself. The key insight enabling this result is the fact that, for spacetime dimensions greater than or equal to three, an arbitrarily small open subset of the boundary of the causal future of a point is sufficient to identify the point uniquely (see Proposition \ref{prp:open2}). Furthermore, the result hinges on the finiteness of the sets $L_i$ and $L_j$ in the definition of conicality. There are issues when examining the definition of span for infinite sets (see Example \ref{ex:counterexample for infinite L}).

\subsection*{Organization}
The paper is organized as follows. In Section~\ref{sec:futurepast}, we introduce basic concepts and notions of causality theory necessary for the proof and introduce the concept conicality. Furthermore, we present spacetime examples that are not conical. In Section~\ref{sec:prep}, we establish a key technical result---an open subset of the boundary of a causal future uniquely determines its generating point. Section~\ref{sec:mainthm} contains the proof of the main theorem, as well as a discussion of what fails when we drop some of the assumptions. Section~\ref{sec:outlook} ends with a summary and suggests some future directions.

\subsection*{Acknowledgment} E. Ling was supported in part by Austrian Science Fund (FWF) [Grant DOI 10.55776/EFP6]. Discussions with large language models were used throughout the preparation of this manuscript. All scientific conclusions, proofs, and final text were independently verified and authored by the researchers, who assume complete responsibility for the entirety of the work.

\section{Conical Spacetimes}\label{sec:futurepast}
Let $(\mathcal{M},g)$ be a Lorentzian manifold with dimension $\operatorname{dim}\mathcal{M}=1+N$ and signature $(-, +,\dots,+)$. We assume both the manifold and metric are smooth.
The metric $g$ induces a causal structure on each tangent space of $\mathcal{M}$ by distinguishing a nonzero vector $X$ as timelike,
null, or spacelike provided $g(X,X)$ is negative, zero, or positive, respectively. A Lorentzian manifold $(\mathcal{M},g)$ is time-oriented if there is a timelike vector field $Y$ on $\mathcal{M}$. In this case, a timelike or null vector $X$ is future or past directed if $g(X,Y)$ is negative or positive, respectively. A time-oriented and connected Lorentzian manifold is referred to as a \emph{spacetime}. 

\subsection{Causal Curves}
Consider a spacetime $(\mathcal{M},g)$.
A piecewise smooth curve \linebreak $\gamma\colon [a,b] \to \mathcal{M}$, with $a < b$ in $\mathbb{R}$,
is called a \emph{future-directed causal curve} if its tangent vector $\dot{\gamma}(t)$ is future-directed and either timelike or null at any of its differentiable points $t$; at any break points, we require this to hold as one-sided limits. Analogous definitions hold for \emph{future-directed timelike} and \emph{null} curves. Likewise \emph{past-directed causal, timelike,} or \emph{null} curves are defined analogously.

\begin{Def}[Causal Relations of Points]
    Given two points $p,q \in \mathcal{M}$, we write $p \prec q$ if there is a future-directed causal curve from $p$ to $q$. We write $p \preceq q$ if $p \prec q$ or $p = q$. 
If there is a future-directed timelike curve from $p$ to $q$, then we write $p \ll q$.\footnote{Our class of curves is not general enough to prove the existence of limit curves. For this, one would need to replace the piecewise smooth assumption with a locally Lipschitz assumption. However, in this paper, we will not make use of limit curves. Moreover, the causal and timelike future sets can be defined with piecewise smooth curves; see~\cite[Corollary 2.4.11]{chrusciel2020geometry}.}
\end{Def}

\begin{Def}[Causal Past \& Future of Sets]
    For a point $p \in \mathcal{M}$, the \emph{causal future} of $p$ is defined as
\begin{equation*}
J^{+}(p) := \{ q \in \mathcal{M} \mid p \preceq q \},
\end{equation*}
and the \emph{causal past} of $p$ is
\begin{equation*}
J^{-}(p) := \{ q \in \mathcal{M} \mid q \preceq p \}.
\end{equation*}

More generally, for a subset $A \subset \mathcal{M}$, the causal future and causal past
are defined respectively by
\begin{align*}
J^{+}(A) &:= \bigcup_{p \in A} J^{+}(p), \\
J^{-}(A) &:= \bigcup_{p \in A} J^{-}(p).
\end{align*}
\end{Def}

Thus, $J^{+}(A)$ consists of all points in spacetime that can be reached from some point
of $A$ by a future-directed causal curve, while $J^{-}(A)$ consists of all points that can
influence $A$ via such curves.

\begin{Def}[Timelike Past \& Future of Sets ]
        For a point $p \in \mathcal{M}$, the \emph{timelike future} of $p$ is defined as
\begin{equation*}
I^{+}(p) := \{ q \in \mathcal{M} \mid p \ll q \},
\end{equation*}
and the \emph{timelike past} of $p$ is
\begin{equation*}
I^{-}(p) := \{ q \in \mathcal{M} \mid q \ll p \}.
\end{equation*}
Analogous definitions hold for $I^+(A)$ and $I^-(A)$ for any subset $A \subset \mathcal{M}$. It's known
that $I^+(A)$ and $I^-(A)$ are open sets in $\mathcal{M}$.
\end{Def}

If $\mathcal{M}'$ is an open subset of a spacetime $\mathcal{M}'$, then we can consider $\mathcal{M}$ as a spacetime itself. We will use the notation $I^+(p, \mathcal{M}')$ for those points $q \in \mathcal{M}'$ which can be reached by $p$ via future-directed timelike curves which are contained in $\mathcal{M}'$. Likewise with $J^+(p, \mathcal{M}')$, etc. 

The following properties hold for any spacetime. A proof can be found, e.g. in the lecture notes~\cite{BaerLorentz}.

\begin{Prp}
Let $(\mathcal{M},g)$ be a spacetime. For any subset $A \subset \mathcal{M}$, the following properties hold.
\begin{itemize}
\item[(i)] ${\rm int}\, J^+(A) = I^+(A)$.
\item[(ii)] The \emph{push-up property}: $x \preceq y \ll z$ implies $x \ll z$. A consequence is: $J^+(A) \subset \overline{I^+(A)}$ with equality if and only if $J^+(A)$ is a closed set. 
\item[(iii)] If $q \in J^+(A) \setminus I^+(A)$, then there is a future-directed null geodesic $\gamma$ starting from some point in $A$ and ending at $q$; moreover $\gamma$ does not enter $I^+(A)$.
\item[(iv)] $A$ is called a \emph{future set} if $I^+(A) \subset A$. If $A$ is a future set, then $\partial A$ is an achronal $C^0$ hypersurface~\cite[Chapter 14, Corollary 27]{o1983semi}. (A set is \emph{achronal} if no two points on it can be connected by a future-directed timelik curve.) Trivial examples of future sets are $I^+(p)$ and $J^+(p)$. 
\end{itemize}
\end{Prp}

We recall some notions from causality theory.

\begin{Def}
Let $(\mathcal{M}, g)$ be a spacetime. 
\begin{itemize}
\item[1.] The spacetime is \emph{causal} if $\gamma(a) \neq \gamma(b)$ for any future-directed causal curve $\g\colon[a,b] \to \mathcal{M}$.

\item[2.] The spacetime is \emph{causally simple} if it is causal and it's causal future $J^+(p)$ and causal past $J^-(p)$ are closed for all $p \in \mathcal{M}$. In this case, $\partial J^+(p) = J^+(p) \setminus I^+(p)$.

\item[3.] The spacetime is \emph{globally hyperbolic} if it's causal and $J^+(p) \cap J^-(q)$ is compact for all $p,q \in \mathcal{M}$.
\end{itemize}
\end{Def}

We have the following chain of implications~\cite{minguzzi2019lorentzian},
\[
\text{Global hyperbolocity} \:\: \Longrightarrow \:\: \text{causal simplicity} \:\: \Longrightarrow \:\: \text{causality}. 
\]

For a causal spacetime $(\mathcal{M},g)$, the causal relation $\preceq$ forms a poset on $\mathcal{M}$. In fact, for any poset, one can define an analogous ``causal future" $J^+$ as above. In~\cite{grothus2024characterizing} the following notion was introduced for any poset.

\begin{Def}[Spanning Points]
\label{def:spanning_points}
Let $\mathcal{T}$ be a poset and fix a nonempty subset $L \subset \mathcal{T}$. The set of \emph{spanning
elements} of $L$, denoted by $\operatorname{span}(L)$, is given by the union of all nonempty sets
$L' \subset L$ that satisfy
\begin{equation*}
JF(L') = JF(L) \: \text{ and }\: \nexists\: \text{ nonempty }\:  L'' \subsetneq L' \: \text{ such that } \: JF(L'') = JF(L),
\end{equation*}
where 
\begin{equation*}
    JF(L) := \bigcap_{x \in L} J^{+}(x)
\end{equation*}
is the \emph{joint future} of the set $L$. We will refer to the nonempty subsets $L'\subset L$ above as \emph{minimal subsets} of $L$. 

If $JF(L) = \emptyset$, then we set $\text{span}(L) = \emptyset$ as convention. We also set $JF(\emptyset) = \mathcal{T}$  since the intersection over an empty family of subsets of $\mathcal{T}$ is $\mathcal{T}$.
\end{Def}

In other words,  $\text{span}(L)$ forms the union of all possible minimal subsets of $L$ which share the same
joint future as $L$. 

\begin{Def}[Conicality]
\label{def:conicality}
Let $\mathcal{T}$ be a poset. We say $\mathcal{T}$ is \emph{conical} if for any
two finite nonempty subsets $L_i, L_j \subset \mathcal{T}$,
\begin{equation*}
\operatorname{span}(L_i) = \operatorname{span}(L_j)
\quad \Longleftrightarrow \quad
JF(L_i) = JF(L_j).
\end{equation*}
\end{Def}

\begin{Remark}
    This definition differs from~\cite[Definition 5.3]{grothus2024characterizing} by requiring an equivalence instead of only the implication $\operatorname{span}(L_i) = \operatorname{span}(L_j)
 \Longrightarrow 
JF(L_i) = JF(L_j)$. Given the text around the definition in~\cite{grothus2024characterizing} we assume that equivalence is required. This also makes sense since the ``$\Longrightarrow$'' direction is a direct implication of the definition of $\operatorname{span}(L)$ (see Proposition \ref{prp: trivial conical implication} below), while the ``$\Longleftarrow$'' direction is the challenging part. 
\end{Remark}

\begin{Prp}\label{prp: trivial conical implication}
Let $\mathcal{T}$ be a poset. If ${\rm span}(L_i) = {\rm span}(L_j)$ for two finite nonempty subsets $L_i,L_j \subset \mathcal{T}$, then $JF(L_i) = JF(L_j)$. 
\end{Prp}

\begin{proof}
First assume 
\[
{\rm span}(L_i)={\rm span}(L_j)\neq\emptyset.
\]
Write
\[
{\rm span}(L_i)=L_i^1\cup\dotsb\cup L_i^n,
\]
where each $L_i^k$ is a minimal subset of $L_i$. Since
$JF(L_i^k)=JF(L_i)$ for every $k$, we have
\[
\begin{aligned}
\bigcap_{x\in {\rm span}(L_i)}J^+(x)
&=
\bigcap_{k=1}^n\bigcap_{x\in L_i^k}J^+(x)\\
&=
\bigcap_{k=1}^n JF(L_i^k)
=
JF(L_i).
\end{aligned}
\]
Similarly,
\[
\bigcap_{x\in {\rm span}(L_j)}J^+(x)=JF(L_j).
\]
Since ${\rm span}(L_i)={\rm span}(L_j)$, it follows that
$JF(L_i)=JF(L_j)$.

Now assume
\[
{\rm span}(L_i)={\rm span}(L_j)=\emptyset.
\]
We claim that $JF(L_i)=JF(L_j)=\emptyset$. Indeed, if, say,
$JF(L_i)\neq\emptyset$, then since $L_i$ is finite and nonempty,
there exists a minimal nonempty subset $L_i'\subseteq L_i$ satisfying
$JF(L_i')=JF(L_i)$. Hence
$L_i'\subseteq {\rm span}(L_i)$, contradicting
${\rm span}(L_i)=\emptyset$. Thus $JF(L_i)=\emptyset$, and the same
argument applies to $L_j$. Therefore $JF(L_i)=JF(L_j)$. (The last argument relies on the finiteness of $L_i$ and $L_j$ and
need not hold for infinite sets; see Example
\ref{ex:counterexample for infinite L}.)
\end{proof}

 The next proposition shows that a necessary condition for conicality is that $L$ has only one minimal subset.
\begin{Prp}
If a poset $\mathcal{T}$ is conical, then any finite nonempty subset $L \subset \mathcal{T}$ with $JF(L) \neq \emptyset$ has a unique minimal subset.
\end{Prp}
\begin{proof}
    Since $JF(L) \neq \emptyset$, we have $\text{span}(L) \neq \emptyset$. Suppose $L$ has two minimal subsets $L'$ and $L''$. Then, by definition, we have $JF(L') = JF(L'')$. Conicality implies $\text{span}(L') = \text{span}(L'')$. However, since $L'$ and $L''$ are minimal, it follows that $\text{span}(L') = L'$ and $\text{span}(L'') = L''$. Hence $L' = L''$.
\end{proof}

Now let's specialize to the case where the poset is given by a causal spacetime $(\mathcal{M},g)$. In essence, conicality for spacetimes means that it is possible to recover the origin of each lightcone contributing to the shape of the joint future from the shape alone.

The following examples show that the implication ``$\Longleftarrow$" in the definition of conicality does not always hold. The first example demonstrates that this property fails in general for two-dimensional spacetimes.

\begin{Example}[\cite{grothus2024characterizing}]\label{ex:2d}
Let $(\mathcal{M},g)$ be two-dimensional Minkowski space: $\mathcal{M} = \mathbb{R}^2$ with canonical coordinates $(t,x)$ and metric $g = -dt^2 + dx^2$. In $(t,x)$ coordinates, consider the points $a = (0,-1)$, $b = (1,0)$, and $c = (0,1)$. Let $L_i = \{a,c\}$ and $L_j = \{b\}$. Then $JF(L_i) = JF(L_j)$ but ${\rm span}(L_i) = L_i \neq L_j = {\rm span}(L_j)$.
\end{Example}

The following example shows that we only want to consider \emph{finite} sets in the definition of conicality.

\begin{Example}\label{ex:counterexample for infinite L}
     Consider Minkowski space in $1+N$ dimensions with $N \geq 2$. Let $\Sigma$ be the Cauchy surface $\Sigma := \{t = 0\}$, and let $L$ be the infinite set $L:=\Sigma\cap \partial J^-(x)$ for some point $x \in I^+(\Sigma)$. Then $JF(L)=J^+ (x)$ and we are back to a similar situation as in Example \ref{ex:2d}. Let $\tilde{\Sigma}$ be a different Cauchy surface that intersects $\partial J^-(x)$ in $\tilde L$. Then $JF(\tilde{L})=JF(L)=J^+ (x)$ but we can always choose $\tilde \Sigma$ such that $L \cap \tilde{L}=\emptyset$. 
     
     However, it gets worse. Recall that the set $\operatorname{span}(L)$ is defined as the union of all minimal sets. However, there is no minimal set of $L$---every dense subset $L'\subset L,$ i.e., $\overline{L'}=L$, satisfies $JF(L')=JF(L)$ and for a subset of $L^*\subset L$ to satisfy $JF(L^*)=JF(L)$ requires $L^*$ to be dense in $L$. And a dense subset of $L$ is never minimal since removing a single point $p$ just generates another dense subset. We therefore conclude that $\operatorname{span}(L') = \operatorname{span}(L)=\emptyset$  but $M= JF (\emptyset)=JF(\operatorname{span}(L))\neq JF(L)=J^+(x)$. Hence the span no-more captures the essential structure of the joint future of the set under consideration.

    From these considerations, we see that Proposition \ref{prp: trivial conical implication} can fail for infinite sets: Let $y \neq x$ be another point in $I^+(\Sigma)$. Set $L_x : = \Sigma \cap \partial J^-(x)$ and $L_y := \Sigma \cap \partial J^-(y)$. Then ${\rm span}(L_x) = \emptyset = {\rm span}(L_y)$ while $JF(L_x) = J^+(x) \neq J^+(y) = JF(L_y)$. 
\end{Example}

Example \ref{ex:counterexample} shows that conicality does not need to hold in higher dimensional spacetimes even if they are globally hyperbolic.
In~\cite{grothus2024characterizing} it was conjectured that causal spacetimes which are homotopic to Minkowski space are conical, provided that the spacetime dimension satisfies $1 + N \geq 3$. Remark \ref{rem: coun-exam} shows how the following example can be modified to give a counterexample to this conjecture.

\begin{Example}\label{ex:counterexample}
This example demonstrates that the full de Sitter space, despite being globally hyperbolic, is not conical. For simplicity, we will work in three-dimensional de Sitter space. Let $(\mathcal{M},g)$ be a portion of the three-dimensional static Einstein universe with $\mathcal{M}\simeq (-\frac{\pi}{2},\frac{\pi}{2}) \times S^2$ and the line element given by 
\begin{equation*}
    ds^2=-dt^2 + d\Omega^2 
\end{equation*}
 where $d\Omega^2$ is the standard line element on $S^2$. (Recall that de Sitter space is conformal to $(\mathcal{M},g)$. Since conformal transformations preserve the causal relation, and hence conicality, it suffices to work with this conformal representation.) Consider an embedding of $S^2$ as the unit sphere into $\mathbb{R}^3$ and consider the spherical coordinate system in the latitude convention, i.e. $(x= \cos{(\theta)}\cos{(\phi)}, y= \cos{\theta}\sin(\phi), z= \sin(\theta))$ with $\phi\in [0,2\pi)$ and $\theta\in (-\frac{\pi}{2},\frac{\pi}{2})$. Let $p_1=(t_0,\tilde p)\in \mathcal{M}$ with $\tilde p= (0,0,1)$, i.e. $\theta=\frac{\pi}{2}$.  The causal future of $p_1$ is then given by $J^+(p_1)=\{q\in \mathcal{M}|t(q)\geq t_0, \theta(q)\geq \frac{\pi}{2}- (t-t_0) \}$. Alternatively in $(t,x,y,z)$ coordinates $J^+(p_1)=\left\{q\in \mathcal{M}|t(q)\geq t_0, z(q)\geq \sin\left(\frac{\pi}{2}- (t-t_0)\right) \right\}$. The boundary of the causal future is then given by $$\partial J^+(p_1)= \left\{q\in \mathcal{M}|t(q)\geq t_0, \theta(q)= \frac{\pi}{2}- (t-t_0) \Leftrightarrow z(q)= \sin\left(\frac{\pi}{2}- (t-t_0)\right) \right\} $$ 

 Next, we consider the joint future of the two antipodal points $p_1=(t_0,0,0,1) $ and  $p_2= (t_0,0,0,-1)$. The joint future of $p_1$ and $p_2$ is then given by $$JF(\{p_1,p_2\})=\left\{q\in \mathcal{M}| t(q)\geq t_0+\frac{\pi}{2}, |z|\leq \sin{\left(t-t_0-\frac{\pi}{2}\right)} \right\}\, .$$ 
 Note that for $t_0\geq 0$ we have that $JF(\{p_1,p_2\})=\emptyset$ because $t\in\left(-\frac{\pi}{2},\frac{\pi}{2}\right)$. At any time $t\geq t_0+\frac{\pi}{2}$ the joint future is given by an interval of  latitudes \[\theta\in\left(\frac{\pi}{2}- (t-t_0), -\frac{\pi}{2}+ (t-t_0)\right),\]
 i.e. a  band around the equator. In the auxiliary cartesian coordinates the band is given by the condition $|z|\leq \sin{\left((t-t_0)-\frac{\pi}{2}\right)}$. To construct a counterexample we therefore need to assume that $t_0<0$. 
 
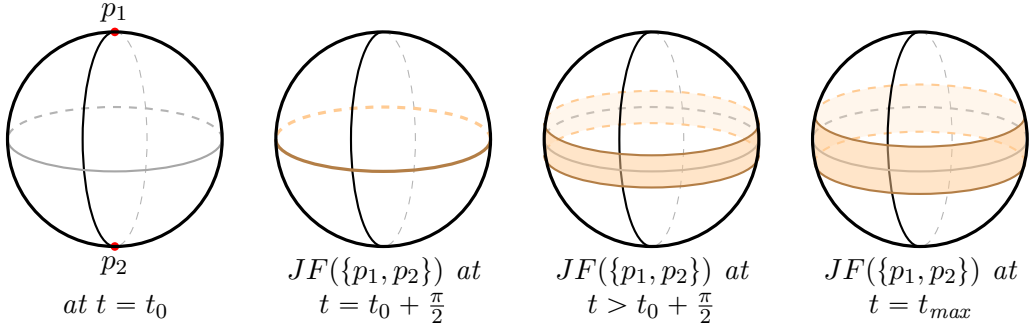
\vspace{0.5cm}

\input{Tikz1}

 We now consider two additional points, $p_3=(t_0,0,1,0)$ and $p_4=(t_0,0,-1,0)$. Hence they are antipodal points in the equatorial plane relative to $p_1,p_2$. The joint future of $p_3,p_4$ is given by $$JF(\{p_3,p_4\})=\left\{q\in \mathcal{M}| t(q)\geq t_0+\frac{\pi}{2}, |y|\leq \sin{\left(t-t_0-\frac{\pi}{2}\right)} \right\}\, .$$ 
 
At any time $t\geq t_0+\frac{\pi}{2}$ this set consists of a band with $|y|\leq \sin{\left((t-t_0)-\frac{\pi}{2}\right)}$. Hence, at any time $t$ the joint futures $JF(\{p_1,p_2\})$ and $JF(\{p_3,p_4\})$  consist of two bands centered around orthogonal planes and we have that
\begin{align*}
    JF(\{p_1,&p_2,p_3,p_4\})=\\ &=\left\{q\in \mathcal{M}| t(q)\geq t_0+\frac{\pi}{2},|y|\leq \sin{\left(t-t_0-\frac{\pi}{2}\right)}, |z|\leq \sin{\left(t-t_0-\frac{\pi}{2}\right)} \right\}\, .
\end{align*} 
 We will now show that for $0<-t_0<\frac{\pi}{4}$ the joint future $JF(\{p_1,p_2,p_3,p_4\})$ has two disconnected components. This simply follows from the fact that $$x^2\geq 1-2\sin{2\left((t-t_0)-\frac{\pi}{2}\right)}^2=\cos{2\left((t-t_0)-\frac{\pi}{2}\right)}$$ and the observation that $\cos{2\left((t-t_0)-\frac{\pi}{2}\right)}>0$ for all $t\in (t_0+\frac{\pi}{2}, \frac{\pi}{2})$ as long as $t<\frac{3\pi}{4}+t_0$ which is satisfied for all $t\in(-\frac{\pi}{2},\frac{\pi}{2} )$ for $t_0\in (-\frac{\pi}{4},0)$. Hence $x^2>0$ for all $t\in (t_0+\frac{\pi}{2}, \frac{\pi}{2})$ and we have that  
 $$JF(\{p_1,p_2,p_3,p_4\})=JF(\{p_1,p_2,p_3,p_4\})_{x_+}\cup JF(\{p_1,p_2,p_3,p_4\})_{x_-}$$ with $JF(\{p_1,p_2,p_3,p_4\})_{x_\pm}$ contained on separate hemispheres of $S^2$ with respect to the plane $x=0$.
 \vspace{0.5cm}

\input{Tikz2}

For a fixed time $t$ the minimum value of $x$ in $JF(\{p_1,p_2,p_3,p_4\})_{x_+}$ is attained in the four corners 
$$\left(t, \sqrt{\cos{2\left((t-t_0)-\frac{\pi}{2}\right)}}, \pm \sin{\left(t-t_0-\frac{\pi}{2}\right)}, \pm\sin{\left(t-t_0-\frac{\pi}{2}\right)}\right)\, $$
and the maximum value of 
$x$ in $JF(\{p_1,p_2,p_3,p_4\})_{x_-}$ is attained in the four corners 
$$\left(t, -\sqrt{\cos{2\left((t-t_0)-\frac{\pi}{2}\right)}}, \pm \sin{\left(t-t_0-\frac{\pi}{2}\right)}, \pm\sin{\left(t-t_0-\frac{\pi}{2}\right)}\right)\, .$$
Now, assume that $t_0 \in (-\frac{\pi}{4},0)$ and hence $JF(\{p_1,p_2,p_3,p_4\})$ consists of two disconnected components $JF(\{p_1,p_2,p_3,p_4\})_{x_+}$ and $JF(\{p_1,p_2,p_3,p_4\})_{x_-}$. 
We want to choose a point $\tilde p$ such that $JF(\{p_1,p_2,p_3,p_4\})_{x_+}\subset J^+(\tilde p)$ but $\partial JF(\{p_1,p_2,p_3,p_4\})_{x_+}\cap \partial J^+(\tilde p)=\emptyset=J^+(\tilde p)\cap JF(\{p_1,p_2,p_3,p_4\})_{x_-}$. 
We will now show that $\tilde p(\tau)= (\tau,1,0,0)$ satisfies these conditions for all $\tau \in \left[ -\frac{\pi}{2}+\arccos{\sqrt{\cos{(2t_0 )}}} , \frac{\pi}{2}-\arccos{\sqrt{\cos{(2t_0 )}}} \right]$.

 The intersection of $J^+(\tilde p(\tau))$ with a constant-$t$ slice is given by
$$ x\geq \cos(t-\tau). $$
Hence for the two conditions to be true we need 
$$ -\sqrt{\cos\!\left(2\left((t-t_0)-\frac{\pi}{2}\right)\right)} < \cos(t-\tau) < \sqrt{\cos\!\left(2\left((t-t_0)-\frac{\pi}{2}\right)\right)}$$
 to hold for all $t\in [t_0+\frac{\pi}{2},\frac{\pi}{2})$. We will now show that the tightest bounds for $\tau$ are obtained in the limit $t\rightarrow \frac{\pi}{2}$.

Since \(t_0\in(-\frac{\pi}{4},0)\) and \(t\in[t_0+\frac{\pi}{2},\frac{\pi}{2})\), the desired inclusion at \(t=t_0+\frac{\pi}{2}\) requires \(\tau<t_0+\frac{\pi}{2}\). Together with \(\tau>-\frac{\pi}{2}\), this implies
$ 0<t-\tau<\pi $
throughout the range under consideration.
 Therefore, since \(\cos\) is strictly decreasing on \((0,\pi)\), introducing
$$ \alpha(t):=\arccos\sqrt{\cos\!\left(2\left((t-t_0)-\frac{\pi}{2}\right)\right)} $$
the above inequality becomes
$$ t-\pi+\alpha(t)<\tau<t-\alpha(t).$$
To show monotonicity we take the derivative with respect to $t$
$$\alpha'(t)
=
\frac{
\sqrt{2}\cos\!\left((t-t_0)-\frac{\pi}{2}\right)
}{
\sqrt{\cos\!\left(2\left((t-t_0)-\frac{\pi}{2}\right)\right)}
}.$$
Since $(t-t_0)-\frac{\pi}{2}\in(0,\pi/4)$ everything is positive and we have ${
\sqrt{2}\cos\!\left((t-t_0)-\frac{\pi}{2}\right)
}>{
\sqrt{\cos\!\left(2\left((t-t_0)-\frac{\pi}{2}\right)\right)}
}$ and therefore $\alpha'(t)>1$. Since $$\frac{d}{dt}\bigl(t-\pi+\alpha(t)\bigr)
=1+\alpha'(t)>0 \,\,\text{ and }\,\, \frac{d}{dt}\bigl(t-\alpha(t)\bigr)
=1-\alpha'(t)<0,$$ the lower bound is strictly increasing, while the upper bound is strictly decreasing. Thus the admissible interval for $\tau$ strictly shrinks as $t\to\frac{\pi}{2}$. Because $t=\frac{\pi}{2}$itself is not in the spacetime and the monotonicity is strict, equality with either limiting bound still gives strict inequalities for every  $t<\frac{\pi}{2}$. 

In summary, for $\tau_1, \tau_2 \in \left[ -\frac{\pi}{2}+\arccos{\sqrt{\cos{(2t_0 )}}} , \frac{\pi}{2}-\arccos{\sqrt{\cos{(2t_0 )}}} \right]$, $\tau_1\neq\tau_2$
 we have that $$JF(\{p_1,p_2,p_3,p_4,\tilde p(\tau_1)\})=JF(\{p_1,p_2,p_3,p_4\})_{x_+}= JF(\{p_1,p_2,p_3,p_4,\tilde p(\tau_2)\})$$ 
 but since each of the points $p_1,p_2,p_3,p_4,\tilde p(\tau)$ contributes one of the defining restrictions above, the sets are minimal by construction and we have 
\begin{align*}
    \operatorname{span}\{p_1,p_2,p_3,p_4,\tilde p(\tau_1)\}&=\{p_1,p_2,p_3,p_4,\tilde p(\tau_1)\}\\&\neq \{p_1,p_2,p_3,p_4,\tilde p(\tau_2)\}=\operatorname{span}\{p_1,p_2,p_3,p_4,\tilde p(\tau_2)\}
\end{align*}
in contradiction to conicality, despite the spacetime being globally hyperbolic.
 
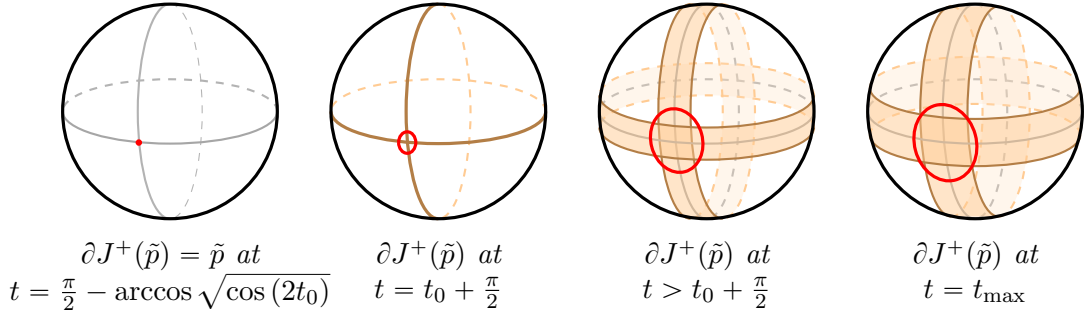
\vspace{0.5cm}

 \input{Tikz3}

\end{Example}
\begin{Remark}\label{rem: coun-exam}
    By removing a point from the sphere $S^2$, and hence a timelike $t$-line in the spacetime, we obtain a non-conical example that's homotopic (in fact homeomorphic) to Minkowski spacetime. This provides a counterexample to~\cite[Conjecture 5.5]{grothus2024characterizing}. Of course, by removing this timelike line from the spacetime, the example is no longer globally hyperbolic, so the conjecture may still hold under global hyperbolocity as suggested in the paragraphs in~\cite{grothus2024characterizing} following the conjecture.
\end{Remark}
In summary the examples show that neither being homotopic to Minkowski space nor being globally hyperbolic alone are sufficient for the spacetime to satisfy conicality. 
In the following, we will prove under which conditions the conjecture does hold true. 

\section{Preparations}\label{sec:prep} 
We begin with the following key observation, which relies crucially on the spacetime dimension satisfying $1+N \geq 3$.

\begin{Prp}\label{prp:open2}
Let $(\mathcal{M},g)$ be a causally simple spacetime of dimension $1 + N \geq 3$. If there is a nonempty $U \subset \partial J^+(p) \cap \partial J^+(q)$ which is open in both $\partial J^+(p)$ and $\partial J^+(q)$, then $p = q$. 
\end{Prp}

\begin{proof}
Fix $x$ in $U$. Let $\g \colon [0,a] \to \mathcal{M}$ be an achronal null geodesic from $p$ to $x$. (Such a $\gamma$ exists by causal simplicity.) Although $x$ can be a null cut point along $\gamma$, every point on $\gamma$ before $x$ is not a null cut point. Since $U$ is open, there is a $y := \g(a')$ in $U$ for some $a' < a$ such that $y$ is not a null cut point of $\gamma|_{[0,a']}$. Let $\lambda \colon [0,b]  \to \mathcal{M}$ be an achronal null geodesic from $q$ to $y$. Either $\g'(a')$ and $\lambda'(b)$ are parallel or they are not. First suppose they are parallel. Then $\gamma$ and $\lambda$ parameterize a portion of the same inextendible null geodesic. Either $p$ is in the future of $q$ or vice versa. Without loss of generality, suppose the latter, $q \in J^+(p)$. Let $\tilde{\lambda}$ denote the future-inextendible null geodesic which extends $\lambda$. If $ q \neq p$, then $J^+(q) \setminus \text{Im}(\tilde{\lambda})$ is contained in $I^+(p)$. Therefore each point in $U$ must lie in the image of $\tilde{\lambda}$, but $\tilde{\lambda}$ is a smooth immersed curve with nowhere vanishing derivative, and the image of such a curve cannot contain\footnote{Indeed, write the parameter interval as a countable union of compact intervals $K_n$. Then $\widetilde\lambda(K_n)$ is contained in a finite union of embedded one-dimensional submanifolds, so $U\cap\widetilde\lambda(K_n)$ is closed in $U$ and has empty interior when $N\ge2$. Since $U$ is a Baire space, it cannot be the countable union of these nowhere-dense sets.} a nonempty open subset of the $N$-dimensional topological manifold $\partial J^+(p)$ with $N \geq 2$. Thus $q = p$. 

Now suppose $\gamma'(a')$ and $\lambda'(b)$ are not parallel.
 Take $z: = \lambda(b')$ in $U$ for some $b' < b$. Let $\sigma \colon [0,c] \to \mathcal{M}$ be an achronal null geodesic from $p$ to $z$. Consider the causal curve $\alpha$ from $p$ to $\gamma(a'')$, for some $a'' \in (a', a)$, which is the concatenation of $\sigma$ with the portion of $\lambda$ from $z$ to $y$ and the concatenation of $\gamma$ from $y$ to $\gamma(a'')$. The curve $\alpha$ is a broken null geodesic. Consequently, there is a timelike curve from $p$ to $\gamma(a'')$, which is a contradiction.
\end{proof}

\input{Tikz4}

\begin{Remark}
    Proposition \ref{prp:open2} fails in $1+1$ dimensions, because the light cone is only one dimensional 
    and hence the last step in the first paragraph of the proof is not true as every open subset $U$ of the boundary $\partial J^+(p)$ contains subsets $\tilde{U}\subseteq U$ that are properly contained in  $\tilde{\g}_q$.

    Many of the arguments below hinge on this one key observation which is the root of the break down of the main theorem in $1+1$ dimensions. Since any point in the $\operatorname{span}(L)$ just contributes a single null geodesic segment to the boundary of $JF(L)$. This is not enough to determine the origin of this light ray uniquely. 
\end{Remark}
Intuitively the statement in Proposition \ref{prp:open2} can be understood as follows: any two open segments of two different null-generators in $\partial J^+(p)$ uniquely identify $p$ through the first intersection in the past of their maximal past extension. We get the following corollary to the above proposition.
\begin{Remark}\label{rem:onedimension}
    A brief way to summarize Proposition \ref{prp:open2} is the following: any open subset of $\partial J^+(p)$ identifies $p$ uniquely. Moreover, as the proof shows, if $p\neq q$ and $q\in J^+(p)$, then $\partial J^+(p)\cap\partial J^+(q)$ is empty, contains only $q$ itself, or it is a one-dimensional subset.
\end{Remark}

% In summery, if $q$ is on the boundaryof the causal future of $p$, then  $\partial J^+(p)\bigcap\partial J^+(q)$ only intersect along the future extension of $\tilde \lambda$ until it runs into a null cut locus.

\section{Proof of the Main Theorem}\label{sec:mainthm}
Here we restate the definitions \ref{def:spanning_points} and \ref{def:conicality} in the language of causality theory for spacetimes. \

In the following we are interested in finite nonempty subsets of spacetime, $L\subset \mathcal{M}$. We wish to characterize these sets in terms of their joint futures,
\begin{equation*}
    JF(L)=\bigcap_{p\in L }J^+(p),
\end{equation*}
or more precisely their spanning set, i.e., the union of all  the smallest nonempty subsets $L'\subset L$ such that the joint future is determined by the subset, i.e. $JF(L')=JF(L)$ and there exists no nonempty subset $L''\subsetneq L'$ such that $JF(L'')=JF(L)$. 
 We want to show under which conditions $JF(L)$ uniquely determines the spanning set. That is, if two finite nonempty subsets $L_i\subset \mathcal{M}$ and $L_j\subset \mathcal{M}$ have the same joint future $JF(L_i)=JF(L_j)$, then do they have the same spanning set, i.e., $\operatorname{span}(L_i)=\operatorname{span}(L_j)$? 

We start by noting the following Lemma. 
\begin{Lemma}\label{lem:spacelike}
    Let $(\mathcal{M},g)$ be a causal spacetime and let $L$ be a finite nonempty subset of $\mathcal{M}$. If $L'\subset L$ is a minimal subset of $L$, then all points in $L'$ are mutually spacelike separated.
\end{Lemma}
\begin{proof}
    Let $L'=\{p_0,\dots, p_n\}$ and assume w.l.o.g. that $p_0\in J^-(p_i)$ for some $i\in [1,n]$. Then we have that $J^+(p_0)\bigcap J^+(p_i)=J^+(p_i)$ and therefore 
    \begin{equation*}
        \bigcap_{i\in [0,n]}J^+(p_i)=\bigcap_{i\in [1,n]}J^+(p_i)
    \end{equation*}
    in contradiction to the assumption that $L'$ is a minimal subset of $L$. (The causality hypothesis is only assumed so that the causal relation $\preceq$ forms a poset on $\mathcal{M}$.)
\end{proof}

\begin{Prp}\label{prp:boundary2}
Let $(\mathcal{M}, g)$ be a causally simple spacetime. Let $L \subset \mathcal{M}$ be finite and nonempty. If $JF(L)$ is nonempty, then $\partial JF(L)$ is a $C^0$ hypersurface. Moreover, for each $q \in \partial JF(L)$, there exists a $p \in L$ such that $q \in \partial J^+(p)$.
\end{Prp}

\begin{proof}
 Note that $JF(L)$ is a future set, i.e., $I^+\big(JF(L)\big) \subset JF(L)$. Therefore $\partial JF(L)$ is  a $C^0$ hypersurface. Moreover, 
\[
\partial JF(L) = \overline{JF(L)} \setminus \text{int}\big(JF(L) \big) = \overline{\bigcap_{p \in L} J^+(p)} \setminus \text{int} \left(\bigcap_{p \in L} J^+(p) \right) = \bigcap_{p \in L} J^+(p) \setminus \bigcap_{p \in L} I^+(p).
\]
In the last equality, we used closedness of $J^+(p)$ and the fact that the interior of a finite intersection is the intersection of the interiors. Thus, for each $q \in \partial JF(L)$, there is some $p \in L$ such that $q \in J^+(p) \setminus I^+(p) = \partial J^+(p)$. 
\end{proof}

We immediately get the following corollary.

\begin{Corollary}\label{prp:boundary}
    In the setting of Proposition \ref{prp:boundary2}, let $q$ be an arbitrary point in $\partial JF(L)$. Then for every minimal subset  $L'\subset L$ there exists an element $p\in L'$ such that $q\in \partial J^+(p)$.
\end{Corollary}
\begin{proof}
    Since $JF(L')=JF(L)$ Proposition \ref{prp:boundary2} applies. 
\end{proof}

Eventually we want to apply Proposition \ref{prp:open2} to identify the points in $\operatorname{span}(L)$ from the joint future $JF(L)$. This requires us to show a stronger result than obtained in the above corollary. In particular, we need to show that for every point in $p\in \operatorname{span}(L)$ the set $\partial JF(L)\cap\partial J^+(p)$ contains an open subset of $\partial J^+(p)$. 
To strengthen Corollary \ref{prp:boundary} accordingly, we need to add an additional assumption given by the following definition.

\begin{Def}\label{def:futureconnected}
    A spacetime $(\mathcal{M}, g)$ is called \emph{future cohesive} if $JF(S)$ is connected for all subsets $S \subset \mathcal{M}$.  
\end{Def}

Definition \ref{def:futureconnected} is stated for all subsets $S \subset \mathcal{M}$, but only the case of finite $S$ is used in what follows; we have kept the stronger form since it is no harder to verify in the examples of interest (see Proposition \ref{prp: TIP connected}).

\begin{Prp}\label{prp: int connected}
Let $(\mathcal{M},g)$ be a spacetime. If $JF(S)$ is connected for some subset $S \subset \mathcal{M}$, then the interior of $JF(S)$ is connected.
\end{Prp}

\begin{proof}
    We prove the contrapositive.
    Assume $\text{int}\,JF(S)$ is disconnected by two disjoint nonempty open sets $A_1$ and $A_2$. Define $B_i := \{x \in JF(S) \mid I^+(x) \subset A_i\}$ for $i = 1,2.$ We show that $B_1$ and $B_2$ are two nonempty disjoint sets which are open in the subspace topology of $JF(S)$ and whose union equals $JF(S)$---contradicting connectedness of $JF(S)$.

    $B_1$ is nonempty since $A_1 \subset B_1$. Indeed, if $x \in A_1$, then, since $JF(S)$ is a future set and $I^+(x)$ is open, we have
    \[
    I^+(x) \subset \text{int}\,JF(S) = A_1 \sqcup A_2.
    \]
    Moreover, $I^+(x) \cap A_1 \neq \emptyset$ because $A_1$ is an open set containing $x$. Since $I^+(x)$ is path-connected, it follows that $I^+(x) \subset A_1$.  Clearly $B_2$ is also nonempty.
    
    $B_1$ is open in $JF(S)$. Indeed, fix \(x\in B_1\) and choose \(y\in I^+(x)\). Then \(y\in A_1\). We claim that
    \[I^-(y)\cap JF(S)\subset B_1.
    \]
    Indeed, if $z\in I^-(y)\cap JF(S)$, then $y\in I^+(z)\cap A_1$. Since $JF(S)$ is a future set,
    \[
    I^+(z)\subset   \operatorname{int}JF(S)=A_1\sqcup A_2. \]
    As $I^+(z)$ is path-connected and meets $A_1$ at $y$, it follows that $I^+(z)\subset A_1$, and hence $z\in B_1$. Thus $I^-(y)\cap JF(S)$ is an open neighborhood of $x$ in $JF(S)$ contained in $B_1$. Clearly $B_2$ is also open in $JF(S)$.

    Now we show $JF(S) = B_1 \sqcup B_2$. Clearly $B_1 \cup B_2 \subset JF(S)$. To show the reverse inclusion, note that $x \in JF(S)$ implies that $I^+(x) \subset \text{int}\, JF(S) = A_1 \sqcup A_2$. Because $I^+(x)$ is path-connected, either $I^+(x) \subset A_1$ or $I^+(x) \subset A_2$. Thus $B_1 \cup B_2 = JF(S)$. Lastly, the fact that $B_1$ and $B_2$ are disjoint follows from $A_1$ and $A_2$ being disjoint. Indeed, if there was a point $x \in B_1 \cap B_2$, then $I^+(x) \subset A_1 \cap A_2 = \emptyset$, contradicting the fact that $I^+(x) \neq \emptyset$.
\end{proof}

Future cohesiveness plays a critical role in the following lemma which states that for any point $\tilde{p}$ in the span, there is a point $q$ in the boundary of the joint future which is not in the boundary of the causal future of any other point in the same minimal subset as $\tilde{p}$. Example \ref{ex:counterexample} shows that future cohesiveness is a necessary assumption for this conclusion.

\begin{Lemma}\label{lem:onlyspanboundary}
Let $(\mathcal{M},g)$ be a causally simple and future cohesive spacetime. Let $L\subset\mathcal{M}$ be a finite nonempty subset. Then for all $\tilde p\in\operatorname{span}(L) $ and $L' \subset L$ a minimal subset  containing $\tilde p$, there exists a point $q\in \partial JF(L)\cap\partial J^+(\tilde p)$ such that $q\notin \partial J^+(p)$  for all $p\in L'\backslash \tilde{p} $.
\end{Lemma}

\begin{proof}
If $JF(L) = \emptyset$, then $\text{span}(L) = \emptyset$ and the statement holds vacuously. 

Assume $JF(L) \neq \emptyset$. Then $JF(L') \neq \emptyset$. The singleton case $L' = \{\tilde p\}$ is trivial so in what follows we assume $L'$ contains at least two points. For notational simplicity, set $A: = J^+(\tilde p)$ and $B := JF(L' \setminus \tilde p)$. Since $L'$ is minimal, we have $JF(L') = A \cap B \subsetneq B$. Also $A = \overline{\text{int} (A)}$ and $B = \overline{\text{int}(B)}$. To prove this, we show $B \subset \overline{\text{int}(B)}$; the reverse inclusion is obvious, and the proof for $A$ is similar. Fix $x \in B$. Then $p \preceq x$ for all $p \in L' \setminus \tilde p$. Let $U$ be a neighborhood of $x$ and fix $y \in I^+(x) \cap U$. Then $p \preceq x \ll y$ which implies $p \ll y$ by the push-up property. Hence $y \in \bigcap_{p \in L' \setminus \tilde p} I^+(p) = \text{int}(B)$. Thus every neighborhood of \(x\) meets \(\operatorname{int}(B)\), and therefore \(x\in\overline{\operatorname{int}(B)}\).

%Therefore every neighborhood $U$ of $x$ intersects $I^+(p)$ for all $p \in L' \setminus \tilde{p}$. Since $\text{int}(B) = \bigcap_{p \in L'\setminus \tilde p}I^+(p)$, we have $x \in \overline{\text{int}(B)}$. 

First, note that $\text{int}(B) \setminus A \neq \emptyset$. Indeed, if it were empty, then $\text{int}(B) \subset A$. Taking the closure gives  $\overline{\text{int}(B)} \subset \overline{A}$, hence $B \subset A$ which contradicts $A \cap B \subsetneq B$. Second, $\text{int}(A \cap B) \neq \emptyset$. Indeed, choose $x \in A\cap B$ and $y \in I^+(x)$. For every $p \in L'$, $p \preceq x \ll y$. Thus $y \in \bigcap_{p \in L'} I^+(p) = \text{int}(A \cap B)$.

Fix $x \in \text{int}(B) \setminus A$ and $y \in \text{int}(A \cap B)$. Since $\text{int}(B)$ is a connected open submanifold by Proposition \ref{prp: int connected}, there is a path $c \colon [0,1] \to \text{int}(B)$ connecting $y = c(0)$ to $x = c(1)$. Define $t_0 := \sup \{t \in [0,1] \mid c\big([0,t]\big) \subset A \cap B\}$. Set $q := c(t_0) \in \partial(A \cap B)$. Since $\partial (A \cap B) \subset (\partial A \cap B) \cup (A \cap \partial B)$ and $q \in \text{int}(B)$, it follows that $q \in \partial A$. Therefore $q \in \partial (A \cap B) \cap \partial A$ and $q \notin \partial B$.

Seeking a contradiction, suppose the conclusion of the lemma is false. Then there is some $p \in L' \setminus \tilde{p}$ such that $q \in \partial J^+(p)$. Therefore any neighborhood $U$ of $q$ contains points that are outside of $J^+(p)$ and hence outside of $B$. Thus $U$ contains points in $B$ (namely $q$) and points outside of $B$, hence $q \in \partial B$, which contradicts the previous paragraph. 
\end{proof}

\begin{Remark}
In fact, the point $q$ obtained in Lemma \ref{lem:onlyspanboundary} is not $\tilde p$. Indeed, if $q = \tilde p$, then $p \ll \tilde p$ for all $p \in L' \setminus \tilde p$  since $q \in {\rm int}(B) = \bigcap_{p \in L'\setminus \tilde p}I^+(p)$, but this would contradict Lemma \ref{lem:spacelike}. 
\end{Remark}

We are now ready to prove the main result of this paper.
\begin{Thm}[Conicality of Causally Simple Future Cohesive Spacetimes]\label{thm: main}
  Let  $(\mathcal{M},g)$ be a causally simple future cohesive spacetime  of dimension $1+N$ with $N\geq 2$.  Then $(\mathcal{M},g)$ is conical according to Definition \ref{def:conicality}.
\end{Thm}
\begin{proof}
One implication was shown in Proposition \ref{prp: trivial conical implication}. Now we show the reverse implication.
Assume $JF(L_1) = JF(L_2) =: J$ for two finite nonempty subsets $L_1, L_2 \subset \mathcal{M}$. We will show $\text{span}(L_1) = \text{span}(L_2)$. First, if $J = \emptyset$, then $\text{span}(L_1) = \text{span}(L_2) = \emptyset$ and the statement holds vacuously. In what follows, we assume $J \neq \emptyset$.  

Fix $p_1 \in \text{span}(L_1)$ and a minimal subset $L_1' \subset L_1$ containing $p_1$. By Lemma \ref{lem:onlyspanboundary}, there is a $q \in \partial J \cap \partial J^+(p_1)$ such that $q \notin \partial J^+(p)$ for all $p \in L'_1 \setminus  p_1$. Since $L_1'$ is finite and the boundaries are closed, there is a neighborhood $V$ of $q$ in the spacetime $\mathcal{M}$ such that $V \cap \partial J^+(p) = \emptyset$ for all $p \in L_1' \setminus p_1$. Moreover, since $q \in I^+(p)$ for all $p \in L_1' \setminus p_1$, we can assume that $V \subset \bigcap_{p \in L_1'\setminus p_1}I^+(p)$. We claim that $V \cap \partial J = V \cap \partial J^+(p_1)$. Indeed, the inclusion, $V \cap \partial J \subset V \cap \partial J^+(p_1)$, follows from the given choice of $V$ and the fact that $\partial J \subset \bigcup_{p \in L_1'} \partial J^+(p)$. To show the reverse inclusion, fix $x \in V \cap \partial J^+(p_1)$. Then $x \in J^+(p)$ for all $p \in L_1'$ and $x \notin I^+(p_1)$ so $x \notin \bigcap_{p \in L_1'}I^+(p) =  \text{int}(J)$, hence $x \in \partial J$. 

Set $W := V \cap \partial J = V \cap \partial J^+(p_1)$ and consider a minimal subset $L_2' \subset L_2$. Since $\partial J \subset \bigcup_{p \in L_2'}\partial J^+(p)$, we have
\[
W = \bigcup_{p \in L_2'}\big(W \cap \partial J^+(p)\big).
\]
Set $C_p := W \cap \partial J^+(p)$. Each $C_p$ is a closed set relative to the subspace topology of $W$. Since $L_2'$ is finite, the topological space $W$ is a finite union of closed sets. Consequently, there must be some $p_2 \in L_2'$ such that $C_{p_2}$ has nonempty interior with respect to the subspace topology of $W$. (A nonempty topological space cannot be a finite union of closed sets each with empty interior.) Therefore, there is a nonempty open set $U$, in the subspace topology of $W$, such that $U \subset \partial J^+(p_2)$. It's clear that $U$ is an open set in $\partial J^+(p_1)$. We show that $U$ is also an open set in $\partial J^+(p_2)$. Indeed, since $U$ is an open set in the topological $N$-manifold $\partial J^+(p_1)$, $U$ itself is a topological $N$-manifold. The inclusion map $U \hookrightarrow \partial J^+(p_2)$ is trivially continuous and injective, so $U$ is also open in $\partial J^+(p_2)$ by Brouwer's theorem on invariance of domain.    

Thus $U$ is open in both $\partial J^+(p_1)$ and  $\partial J^+(p_2)$; hence $p_1 = p_2$ by Proposition \ref{prp:open2}. Therefore $p_1 \in \text{span}(L_2)$. The reverse inclusion, $\text{span}(L_2) \subset \text{span}(L_1)$, follows by a symmetrical argument. 
\end{proof}

\begin{Remark}
It is worth noting where we used finiteness of $L$ in the proof of the main theorem. First, we used it in Proposition \ref{prp:boundary2}. Second, we used it in the construction of the neighborhood $V$ in the proof of Theorem \ref{thm: main}. Third, we used it to show that $C_{p_2}$ has nonempty interior. 
\end{Remark}

\begin{Remark}
In the following we summarize a ``practical'' algorithm for recovering $\operatorname{span}(L)$ given a nonempty $JF(L)$
%\footnote{If $JF(L)=\emptyset$ we have $\operatorname{span}(L)=\emptyset$ and the procedure is trivial. }
in a causally simple and future cohesive spacetime. 
\begin{enumerate}
    \item As a first step we identify all null geodesics $\gamma$ in $\mathcal{M}$ such that there exists an open interval $(s_1,s_2)$ with $\gamma(s)\in\partial JF(L)$ for all $s\in (s_1, s_2)$. We call the set of all such null geodesics $\Gamma(JF(L))$.
    \item Determine the maximal extension of all null geodesics in the set $\Gamma(JF(L))$. Let $\gamma:(a,b)\rightarrow\mathcal{M}$ be the maximal extension of an element in $\Gamma(JF(L))$ then we define $s_{\text{min}}(\gamma):=\min\{s\in(a,b)\;|\; \gamma(s)\in \partial JF(L)\}$.
    \item We then have that the span of $L$ is the set of points where maximal extensions of elements in $\Gamma$ first intersect the maximal extension of another element in $\Gamma$ in the past of $JF(L)$
    \begin{equation*}
        \operatorname{span}(L):= \bigcup_{\gamma\in\Gamma} \Big\{\gamma(s_{\text{max}})\Big| s_{\text{max}}:=\max\{ s\in(a,s_{\text{min}})| \exists \tilde{\gamma}\in \Gamma \text{ such that } \gamma(s)\in \tilde{\gamma}\}\Big\}.
    \end{equation*}
\end{enumerate}
\end{Remark}

The next proposition shows that the past of any observer, i.e. $I^-(\gamma)$ for an future-inextendible timelike curve, which is the natural domain for causal modeling, is future cohesive. Sets of the form $I^-(\gamma)$ are often referred to as TIPs~\cite{hawking2023large}.

\begin{Prp}\label{prp: TIP connected}
Let $(\mathcal{M}, g)$ be a spacetime. Let $\gamma$ be a future-directed future-inextendible timelike curve. Then $I^-(\gamma)$, as a spacetime, is future cohesive.
\end{Prp}

\begin{proof}
Set $\mathcal{M}' := I^-(\gamma)$. Fix a subset $S \subset \mathcal{M}'$. The joint future of $S$ within $\mathcal{M}'$ is  $JF(S, \mathcal{M}') = \bigcap_{p \in S} J^+(p, \mathcal{M}')$. Fix $x,y \in JF(S,\mathcal{M}')$. There are points $a,b$ on the image of $\gamma$ such that $x \ll a$ and $y \ll b$. Without loss of generality, we can assume $a \ll b$ (i.e., $a$ precedes $b$ on $\gamma$). Then there are timelike curves $\alpha$ and $\beta$ from $x$ to $a$ and $y$ to $b$, respectively. Let $\lambda$ denote the portion of $\gamma$ from $a$ to $b$. Then there is a continuous curve from $x$ to $y$ given by the concatenation of $\alpha$, $\lambda$, and $-\beta$.  By construction, $\alpha$, $\beta$, and $\gamma$ all belong to $JF(S,\mathcal{M}')$. Thus $JF(S, \mathcal{M}')$ is path connected. 
\end{proof}

\begin{Corollary}\label{cor: tip}
    Let $(\mathcal{M},g)$ be a causally simple spacetime of dimension $\geq 3$. 
    Let $\gamma$ be a future-directed future-inextendible timelike curve. Then $I^-(\gamma)$, as a spacetime, is conical. 
\end{Corollary}

\begin{proof}
It suffices to show that $I^-(\gamma)$, as a spacetime, is causally simple. For then the result follows from Theorem \ref{thm: main} and Proposition \ref{prp: TIP connected}. Set $\mathcal{M}' := I^-(\gamma)$. Clearly $\mathcal{M}'$ is causal for otherwise $\mathcal{M}$ would not be. We show $J^+(p, \mathcal{M}')$ is closed. In fact, we show $J^+(p, \mathcal{M}') = J^+(p) \cap \mathcal{M}'$. Clearly $J^+(p, \mathcal{M}') \subset J^+(p) \cap \mathcal{M}'$. Conversely, suppose $q \in J^+(p) \cap \mathcal{M}'$. Let $\alpha$ be a future-directed causal curve from $p$ to $q$. Since $\mathcal{M}'$ is a past set (i.e., $I^-(\mathcal{M}') = \mathcal{M}')$, it follows that $\alpha \subset \mathcal{M}'$ from an application of the push-up property. Thus $q \in J^+(p, \mathcal{M}')$. The proof that $J^-(p, \mathcal{M}')$ is closed is analogous.
\end{proof}

\begin{Remark}
    Note that a finite slice of Minkowski space (i.e. $\mathcal{M} = \{a < t < b\}$) is future cohesive, but is not a TIP. Hence, conicality holds in a class of spacetimes that is strictly larger than the set of all causally simple TIPs. 
\end{Remark}

\section{Summary and Outlook}\label{sec:outlook}

In~\cite{grothus2024characterizing}, the authors introduced the notion of \emph{conicality} (Definition \ref{def:conicality}) for posets and raised the question: which spacetimes are conical? Briefly, conicality means that the span of a set can always be recovered uniquely from the joint future of that set.

In this paper, we prove Theorem \ref{thm: main} demonstrating that causally simple and future cohesive (Definition \ref{def:futureconnected}) spacetimes are conical. In particular, all TIPs in a causally simple spacetime are conical (Corollary \ref{cor: tip}). 

It is worth emphasizing that the two hypotheses of Theorem \ref{thm: main} are of a rather different nature. Causal simplicity sits high on the causal ladder, just below global hyperbolicity~\cite{minguzzi2008causal}, and is a genuine causality condition: it constrains the global behavior of causal curves. Future cohesiveness, by contrast, is not a causality condition and is logically independent of the standard conditions on the causal ladder. Indeed, de Sitter space is globally hyperbolic but not future cohesive, whereas the flat Lorentzian torus is future cohesive, since every joint future is the entire spacetime, but is not causal. 
Rather, future cohesiveness is a connectedness requirement on joint futures, closer in spirit to a topological condition on causal future sets than to one of the usual conditions on the causal ladder. 
This difference is reflected in where the two hypotheses enter the proof. Causal simplicity is used throughout, but its main application is to prove the rigidity statement in Proposition \ref{prp:open2}. Future cohesiveness is used only in Lemma \ref{lem:onlyspanboundary}, which guarantees that each element of the span contributes a point on the boundary of the joint future lying on no other generator's light cone---an essential step in the proof of Theorem \ref{thm: main}.

Several directions can be explored. First, the proof of Theorem \ref{thm: main} relies on tools that are specific to smooth Lorentzian geometry---achronal null geodesics, null cut points, and the $C^0$ hypersurface structure of achronal boundaries---and it is natural to ask how much of it survives in lower regularity settings as in~\cite{Chrusciel_2012,Ling2020,Minguzzi_cone}. For $C^{1,1}$ metrics, much of causality theory persists~\cite{Kunzinger2014,Minguzzi2015}, but for metrics of lower regularity, the rigidity in Proposition \ref{prp:open2} may require modification if geodesics are not unique. More broadly, conicality is formulated solely in terms of posets, so it makes sense to pose the same question for synthetic Lorentzian structures such as Lorentzian (pre)length spaces and causal sets~\cite{kunzinger2018lorentzian,Hau_2020,Minguzzi2024,Ling2025,minguzzi2026globalhyperbolicitymeetsorder,PhysRevLett.59.521,Dowker01012006,surya2011directionscausalsetquantum}. Identifying the right substitute for Proposition \ref{prp:open2} in such a framework appears to be the essential obstruction, and its formulation would clarify which features of the smooth setting conicality genuinely requires.

A second direction concerns stability. Conicality says that $JF(L) \mapsto \text{span}(L)
$ is a well-defined map, but it does not say anything about its regularity. In particular, one would like to know if it is continuous: if two finite sets have joint futures that are close, are their spans close as well? Making this precise requires a topology on each side of the map, and the obvious choice does not work. Indeed, it is not hard to give examples in Minkowski spacetime showing that the map is not continuous with respect to Hausdorff distance: let $p$ be the origin, $q \neq p$ be on the lightcone of $p$, and $q_k \to q$ in the timelike past of $q$. Then $JF(p,q_k)$ approaches $JF(p,q) = J^+(q)$, while $\text{span}(p,q_k) = \{p,q_k\}$ and $\text{span}(p,q) =\{q\}$. The Hausdorff distance between $\{p,q_k\}$ and $\{q\}$, with respect to the canonical Euclidean metric $\delta$, is equal to $\delta(p,q) > 0$ for all sufficiently large $k$. Whether some other topology restores continuity, or whether continuity should instead be sought on a restricted class of configurations, remains to be seen.

\printbibliography

\end{document}

%% file: Tikz1.tex
\begin{tikzpicture}[scale=1.42]
    % Define colors
    \definecolor{spherecolor}{RGB}{100,150,200}
    \definecolor{bandcolor}{RGB}{255,180,100}
    
    % === FIRST SPHERE: No band, grey equator ===
    \begin{scope}[xshift=0cm, rotate=180]
        % Back meridian (dashed)
        \draw[dashed, gray!60] (0,-1) arc (270:90:0.3cm and 1cm);
        
        % Back part of equator (dashed)
        \draw[dashed, gray!60, thick] (-1,0) arc (180:360:1cm and 0.3cm);
        
        % Front meridian (solid)
        \draw[thick] (0,-1) arc (270:90:-0.3cm and 1cm);
        
        % Front part of equator (solid, gray)
        \draw[thick, gray!70] (-1,0) arc (180:0:1cm and 0.3cm);

%red dot%red dot
  \fill[red] (0,-1) circle(0.04cm);
  \fill[red] (0,1) circle(0.04cm);

        % Outer circle (limb of sphere)
        \draw[very thick] (0,0) circle (1cm);
        \node[above] at (0,-1){$p_1$};
        \node[below, align=center] at (0,1) {$p_2$};
        \node[below, align=center] at (0,1) {\vphantom{$JF(\{p_1,p_2\})$} \\ at $t=t_0$};p
    \end{scope}
    
    % === SECOND SPHERE: Thick orange equator, no band ===
    \begin{scope}[xshift=2.5cm, rotate=180]
        % Back meridian (dashed)
        \draw[dashed, gray!60] (0,-1) arc (270:90:0.3cm and 1cm);
        
        % Back part of equator (dashed, orange)
        \draw[dashed, bandcolor!70, very thick] (-1,0) arc (180:360:1cm and 0.3cm);
        
        % Front meridian (solid)
        \draw[thick] (0,-1) arc (270:90:-0.3cm and 1cm);
        
        % Front part of equator (solid, thick orange)
        \draw[very thick, bandcolor!70!black] (-1,0) arc (180:0:1cm and 0.3cm);
        
        % Outer circle (limb of sphere)
        \draw[very thick] (0,0) circle (1cm);
        \node[below, align=center] at (0,1) {$JF(\{p_1,p_2\})$ at \\ $t=t_0+\frac{\pi}{2}$};
    \end{scope}
    
    % === THIRD SPHERE: Like prototype ===
    \begin{scope}[xshift=5cm, rotate=180]
        % Back meridian (dashed)
        \draw[dashed, gray!60] (0,-1) arc (270:90:0.3cm and 1cm);
        
        % Back part of equator (dashed)
        \draw[dashed, gray!60, thick] (-1,0) arc (180:360:1cm and 0.3cm);
        
        % Back part of band
        \fill[bandcolor!40, opacity=0.3] 
            (-1,0.15) arc (180:360:1cm and 0.3cm) -- 
            (1,-0.15) arc (360:180:1cm and 0.3cm) -- cycle;
        
        % Back edges of band (dashed)
        \draw[dashed, bandcolor!70, thick] (-1,0.15) arc (180:360:1cm and 0.3cm);
        \draw[dashed, bandcolor!70, thick] (-1,-0.15) arc (180:360:1cm and 0.3cm);
        
        % Front meridian (solid)
        \draw[thick] (0,-1) arc (270:90:-0.3cm and 1cm);
        
        % Front part of equator (solid, gray)
        \draw[thick, gray!70] (-1,0) arc (180:0:1cm and 0.3cm);
        
        % Front part of band
        \fill[bandcolor!70, opacity=0.5] 
            (-1,0.15) arc (180:0:1cm and 0.3cm) -- 
            (1,-0.15) arc (0:180:1cm and 0.3cm) -- cycle;
        
        % Front edges of band (solid)
        \draw[bandcolor!70!black, thick] (-1,0.15) arc (180:0:1cm and 0.3cm);
        \draw[bandcolor!70!black, thick] (-1,-0.15) arc (180:0:1cm and 0.3cm);
        
        % Outer circle (limb of sphere)
        \draw[very thick] (0,0) circle (1cm);

        \node[below, align=center] at (0,1) {$JF(\{p_1,p_2\})$ at \\ $t>t_0+\frac{\pi}{2}$};
    \end{scope}
    
    % === FOURTH SPHERE: Like prototype but thicker band ===
    \begin{scope}[xshift=7.5cm, rotate=180]
        % Back meridian (dashed)
        \draw[dashed, gray!60] (0,-1) arc (270:90:0.3cm and 1cm);
        
        % Back part of equator (dashed)
        \draw[dashed, gray!60, thick] (-1,0) arc (180:360:1cm and 0.3cm);
        
        % Back part of band (thicker - using 0.22 to stay within sphere)
        \fill[bandcolor!40, opacity=0.3] 
            (-0.98,0.22) arc (180:360:0.98cm and 0.29cm) -- 
            (0.98,-0.22) arc (360:180:0.98cm and 0.29cm) -- cycle;
        
        % Back edges of band (dashed, thicker)
        \draw[dashed, bandcolor!70, thick] (-0.98,0.22) arc (180:360:0.98cm and 0.29cm);
        \draw[dashed, bandcolor!70, thick] (-0.98,-0.22) arc (180:360:0.98cm and 0.29cm);
        
        % Front meridian (solid)
        \draw[thick] (0,-1) arc (270:90:-0.3cm and 1cm);
        
        % Front part of equator (solid, gray)
        \draw[thick, gray!70] (-1,0) arc (180:0:1cm and 0.3cm);
        
        % Front part of band (thicker)
        \fill[bandcolor!70, opacity=0.5] 
            (-0.98,0.22) arc (180:0:0.98cm and 0.29cm) -- 
            (0.98,-0.22) arc (0:180:0.98cm and 0.29cm) -- cycle;
        
        % Front edges of band (solid, thicker)
        \draw[bandcolor!70!black, thick] (-0.98,0.22) arc (180:0:0.98cm and 0.29cm);
        \draw[bandcolor!70!black, thick] (0.98,-0.22) arc (0:180:0.98cm and 0.29cm);
        
        % Outer circle (limb of sphere)
        \draw[very thick] (0,0) circle (1cm);
\node[below, align=center] at (0,1) {$JF(\{p_1,p_2\})$ at \\ $t=t_{\text{max}}$};
        
    \end{scope}
    
\end{tikzpicture}

%% file: Tikz2.tex
% Requires: amsmath for \text command
\begin{tikzpicture}[scale=1.42]
    % Define colors
    \definecolor{spherecolor}{RGB}{100,150,200}
    \definecolor{bandcolor}{RGB}{255,180,100}
    
    % === FIRST SPHERE: No band, grey equator ===
    \begin{scope}[xshift=0cm, rotate=180]
        % Back meridian (dashed)
        \draw[dashed, gray!60] (0,-1) arc (270:90:0.3cm and 1cm);
        
        % Back part of equator (dashed)
        \draw[dashed, gray!60, thick] (-1,0) arc (180:360:1cm and 0.3cm);
        
        % Front meridian (solid)
        \draw[thick, gray!60] (0,-1) arc (270:90:-0.3cm and 1cm);
        
        % Front part of equator (solid, gray)
        \draw[thick, gray!70] (-1,0) arc (180:0:1cm and 0.3cm);
     %red dot%red dot
  \fill[red] (0,-1) circle(0.04cm);
  \fill[red] (0,1) circle(0.04cm);
  \fill[red] (-1,0) circle(0.04cm);
  \fill[red] (1,0) circle(0.04cm);
        
        % Outer circle (limb of sphere)
        \draw[very thick] (0,0) circle (1cm);
        \node[right] at (-1,0){$p_3$};
         \node[left] at (1,0){$p_4$};
        \node[above] at (0,-1){$p_1$};
        \node[below, align=center] at (0,1) {$p_2$};
        \node[below, align=center] at (0,1.3){\vphantom{$JF(\{p_1,p_2\})$}\\ \vphantom{$JF(\{p_1,p_2\})$}\\ at $t=t_0$};
    \end{scope}
    
    % === SECOND SPHERE: Thick orange equator + meridian band ===
    \begin{scope}[xshift=2.5cm, rotate=180]
        % Back part of equator (dashed, orange)
        \draw[dashed, bandcolor!70, thick] (-1,0) arc (180:360:1cm and 0.3cm);
        
        % Back meridian (dashed)
        \draw[dashed, bandcolor!70, thick] (0,-1) arc (270:90:0.3cm and 1cm);
        
        % Front meridian (solid)
        \draw[bandcolor!70!black,very thick] (0,-1) arc (270:90:-0.3cm and 1cm);
        
        % Front part of equator (solid, thick orange)
        \draw[very thick, bandcolor!70!black] (-1,0) arc (180:0:1cm and 0.3cm);
        
        % Outer circle (limb of sphere)
        \draw[very thick] (0,0) circle (1cm);
        \node[below, align=center] at (0,1.3) {$JF(\{p_1,p_2\})$ and \\ $JF(\{p_3,p_4\})$ at \\ $t=t_0+\frac{\pi}{2}$};
    \end{scope}
    
    % === THIRD SPHERE: Equator band + meridian band ===
    \begin{scope}[xshift=5cm, rotate=180]
        % Back part of equator (dashed)
        \draw[dashed, gray!60, thick] (-1,0) arc (180:360:1cm and 0.3cm);
        % Back meridian (dashed)
        \draw[dashed, gray!60, thick] (0,-1) arc (270:90:0.3cm and 1cm);
        
        % Back part of equator band
        \fill[bandcolor!40, opacity=0.3] 
            (-1,0.15) arc (180:360:1cm and 0.3cm) -- 
            (1,-0.15) arc (360:180:1cm and 0.3cm) -- cycle;
        
        % Back edges of equator band (dashed)
        \draw[dashed, bandcolor!70, thick] (-1,0.15) arc (180:360:1cm and 0.3cm);
        \draw[dashed, bandcolor!70, thick] (-1,-0.15) arc (180:360:1cm and 0.3cm);
        
        % Back meridian band
        \fill[bandcolor!40, opacity=0.3] 
            (-0.15,-1) arc (270:90:0.3cm and 1cm) -- 
            (0.15,1) arc (90:270:0.3cm and 1cm) -- cycle;
        
        % Back edges of meridian band (dashed)
        \draw[dashed, bandcolor!70, thick] (-0.15,-1) arc (270:90:0.3cm and 1cm);
        \draw[dashed, bandcolor!70, thick] (0.15,-1) arc (270:90:0.3cm and 1cm);
% Front meridian (solid)
        \draw[gray!70, thick] (0,-1) arc (270:90:-0.3cm and 1cm);
        
        % Front meridian band
        \fill[bandcolor!70, opacity=0.5] 
            (-0.15,-1) arc (270:90:-0.3cm and 1cm) -- 
            (0.15,1) arc (90:270:-0.3cm and 1cm) -- cycle;

        % Front edges of meridian band (solid)
        \draw[bandcolor!70!black, thick] (-0.15,-1) arc (270:90:-0.3cm and 1cm);
        \draw[bandcolor!70!black, thick] (0.15,-1) arc (270:90:-0.3cm and 1cm);
        
        % Front part of equator (solid, gray)
        \draw[thick, gray!70] (-1,0) arc (180:0:1cm and 0.3cm);
        
        % Front part of equator band
        \fill[bandcolor!70, opacity=0.5] 
            (-1,0.15) arc (180:0:1cm and 0.3cm) -- 
            (1,-0.15) arc (0:180:1cm and 0.3cm) -- cycle;
        
        % Front edges of equator band (solid)
        \draw[bandcolor!70!black, thick] (-1,0.15) arc (180:0:1cm and 0.3cm);
        \draw[bandcolor!70!black, thick] (-1,-0.15) arc (180:0:1cm and 0.3cm);
        
        % Outer circle (limb of sphere)
        \draw[very thick] (0,0) circle (1cm);
        \node[below, align=center] at (0,1.3) {$JF(\{p_1,p_2\})$ and \\ $JF(\{p_3,p_4\})$ at \\ $t>t_0+\frac{\pi}{2}$};
    \end{scope}
    
    % === FOURTH SPHERE: Thicker bands ===
    \begin{scope}[xshift=7.5cm, rotate=180]
        % Back part of equator (dashed)
        \draw[dashed, gray!60, thick] (-1,0) arc (180:360:1cm and 0.3cm);
         % Back meridian (dashed)
        \draw[dashed, gray!60, thick] (0,-1) arc (270:90:0.3cm and 1cm);
        
        % Back part of equator band (thicker)
        \fill[bandcolor!40, opacity=0.3] 
            (-0.98,0.22) arc (180:360:0.98cm and 0.29cm) -- 
            (0.98,-0.22) arc (360:180:0.98cm and 0.29cm) -- cycle;
        
        % Back edges of equator band (dashed, thicker)
        \draw[dashed, bandcolor!70, thick] (-0.98,0.22) arc (180:360:0.98cm and 0.29cm);
        \draw[dashed, bandcolor!70, thick] (-0.98,-0.22) arc (180:360:0.98cm and 0.29cm);
        
        % Back meridian band (thicker)
        \fill[bandcolor!40, opacity=0.3] 
            (-0.22,-0.98) arc (270:90:0.29cm and 0.98cm) -- 
            (0.22,0.98) arc (90:270:0.29cm and 0.98cm) -- cycle;
        
        % Back edges of meridian band (dashed, thicker)
        \draw[dashed, bandcolor!70, thick] (-0.22,-0.98) arc (270:90:0.29cm and 0.98cm);
        \draw[dashed, bandcolor!70, thick] (0.22,-0.98) arc (270:90:0.29cm and 0.98cm);

% Front meridian (solid)
        \draw[gray!70, thick] (0,-1) arc (270:90:-0.3cm and 1cm);
        
        % Front meridian band (thicker)
        \fill[bandcolor!70, opacity=0.5] 
            (-0.22,-0.98) arc (270:90:-0.29cm and 0.98cm) -- 
            (0.22,0.98) arc (90:270:-0.29cm and 0.98cm) -- cycle;
        
        % Front edges of meridian band (solid, thicker)
        \draw[bandcolor!70!black, thick] (-0.22,-0.98) arc (270:90:-0.29cm and 0.98cm);
        \draw[bandcolor!70!black, thick] (0.22,-0.98) arc (270:90:-0.29cm and 0.98cm);
        
        % Front part of equator (solid, gray)
        \draw[thick, gray!70] (-1,0) arc (180:0:1cm and 0.3cm);
        
        % Front part of equator band (thicker)
        \fill[bandcolor!70, opacity=0.5] 
            (-0.98,0.22) arc (180:0:0.98cm and 0.29cm) -- 
            (0.98,-0.22) arc (0:180:0.98cm and 0.29cm) -- cycle;
        
        % Front edges of equator band (solid, thicker)
        \draw[bandcolor!70!black, thick] (-0.98,0.22) arc (180:0:0.98cm and 0.29cm);
        \draw[bandcolor!70!black, thick] (0.98,-0.22) arc (0:180:0.98cm and 0.29cm);
        
        % Outer circle (limb of sphere)
        \draw[very thick] (0,0) circle (1cm);
        \node[below, align=center] at (0,1.3) {$JF(\{p_1,p_2\})$ and \\ $JF(\{p_3,p_4\})$ at \\ $t=t_{\mathrm{max}}$};
    \end{scope}
    
\end{tikzpicture}

%% file: Tikz3.tex
% Requires: amsmath for \text command
\hspace{-0.4cm}\begin{tikzpicture}[scale=1.42]
    % Define colors
    \definecolor{spherecolor}{RGB}{100,150,200}
    \definecolor{bandcolor}{RGB}{255,180,100}
    
    % === FIRST SPHERE: No band, grey equator ===
    \begin{scope}[xshift=0cm, rotate=180]
        % Back meridian (dashed)
        \draw[dashed, gray!60] (0,-1) arc (270:90:0.3cm and 1cm);
        
        % Back part of equator (dashed)
        \draw[dashed, gray!60, thick] (-1,0) arc (180:360:1cm and 0.3cm);
        
        % Front meridian (solid)
        \draw[thick, gray!60] (0,-1) arc (270:90:-0.3cm and 1cm);
        
        % Front part of equator (solid, gray)
        \draw[thick, gray!70] (-1,0) arc (180:0:1cm and 0.3cm);

%red dot
  \draw[very thick,red] (0.29,0.29) circle(0.015cm);

        % Outer circle (limb of sphere)
        \draw[very thick] (0,0) circle (1cm);
        \node[below, align=center] at (0,1.1) {$\partial J^+(\tilde p)=\tilde p$ at \\ $t=\frac{\pi}{2}-\arccos{\sqrt{\cos{(2t_0 )}}}$};
    \end{scope}
    
    % === SECOND SPHERE: Thick orange equator + meridian band ===
    \begin{scope}[xshift=2.5cm, rotate=180]
        % Back part of equator (dashed, orange)
        \draw[dashed, bandcolor!70, thick] (-1,0) arc (180:360:1cm and 0.3cm);
        
        % Back meridian (dashed)
        \draw[dashed, bandcolor!70, thick] (0,-1) arc (270:90:0.3cm and 1cm);
        
        % Front meridian (solid)
        \draw[bandcolor!70!black,very thick] (0,-1) arc (270:90:-0.3cm and 1cm);
        
        % Front part of equator (solid, thick orange)
        \draw[very thick, bandcolor!70!black] (-1,0) arc (180:0:1cm and 0.3cm);

%red circle
  \draw[very thick,red] (0.37,0.29) arc (0:360:0.08cm and 0.1cm);
        
        % Outer circle (limb of sphere)
        \draw[very thick] (0,0) circle (1cm);
        \node[below, align=center] at (0,1.1) {$\partial J^+(\tilde p)$  at \\ $t=t_0+\frac{\pi}{2}$};
    \end{scope}
    
    % === THIRD SPHERE: Equator band + meridian band ===
    \begin{scope}[xshift=5cm, rotate=180]
        % Back part of equator (dashed)
        \draw[dashed, gray!60, thick] (-1,0) arc (180:360:1cm and 0.3cm);
        % Back meridian (dashed)
        \draw[dashed, gray!60, thick] (0,-1) arc (270:90:0.3cm and 1cm);
        
        % Back part of equator band
        \fill[bandcolor!40, opacity=0.3] 
            (-1,0.15) arc (180:360:1cm and 0.3cm) -- 
            (1,-0.15) arc (360:180:1cm and 0.3cm) -- cycle;
        
        % Back edges of equator band (dashed)
        \draw[dashed, bandcolor!70, thick] (-1,0.15) arc (180:360:1cm and 0.3cm);
        \draw[dashed, bandcolor!70, thick] (-1,-0.15) arc (180:360:1cm and 0.3cm);
        
        % Back meridian band
        \fill[bandcolor!40, opacity=0.3] 
            (-0.15,-1) arc (270:90:0.3cm and 1cm) -- 
            (0.15,1) arc (90:270:0.3cm and 1cm) -- cycle;
        
        % Back edges of meridian band (dashed)
        \draw[dashed, bandcolor!70, thick] (-0.15,-1) arc (270:90:0.3cm and 1cm);
        \draw[dashed, bandcolor!70, thick] (0.15,-1) arc (270:90:0.3cm and 1cm);
% Front meridian (solid)
        \draw[gray!70, thick] (0,-1) arc (270:90:-0.3cm and 1cm);
        
        % Front meridian band
        \fill[bandcolor!70, opacity=0.5] 
            (-0.15,-1) arc (270:90:-0.3cm and 1cm) -- 
            (0.15,1) arc (90:270:-0.3cm and 1cm) -- cycle;

        % Front edges of meridian band (solid)
        \draw[bandcolor!70!black, thick] (-0.15,-1) arc (270:90:-0.3cm and 1cm);
        \draw[bandcolor!70!black, thick] (0.15,-1) arc (270:90:-0.3cm and 1cm);
        
        % Front part of equator (solid, gray)
        \draw[thick, gray!70] (-1,0) arc (180:0:1cm and 0.3cm);
        
        % Front part of equator band
        \fill[bandcolor!70, opacity=0.5] 
            (-1,0.15) arc (180:0:1cm and 0.3cm) -- 
            (1,-0.15) arc (0:180:1cm and 0.3cm) -- cycle;
        
        % Front edges of equator band (solid)
        \draw[bandcolor!70!black, thick] (-1,0.15) arc (180:0:1cm and 0.3cm);
        \draw[bandcolor!70!black, thick] (-1,-0.15) arc (180:0:1cm and 0.3cm);

        %red circle
  \draw[very thick,red,rotate=15] (0.58,0.19) arc (0:360:0.24cm and 0.3cm);
        
        % Outer circle (limb of sphere)
        \draw[very thick] (0,0) circle (1cm);

        \node[below, align=center] at (0,1.1) {$\partial J^+(\tilde p)$ at \\ $t>t_0+\frac{\pi}{2}$};
    \end{scope}
    
    % === FOURTH SPHERE: Thicker bands ===
    \begin{scope}[xshift=7.5cm, rotate=180]
        % Back part of equator (dashed)
        \draw[dashed, gray!60, thick] (-1,0) arc (180:360:1cm and 0.3cm);
         % Back meridian (dashed)
        \draw[dashed, gray!60, thick] (0,-1) arc (270:90:0.3cm and 1cm);
        
        % Back part of equator band (thicker)
        \fill[bandcolor!40, opacity=0.3] 
            (-0.98,0.22) arc (180:360:0.98cm and 0.29cm) -- 
            (0.98,-0.22) arc (360:180:0.98cm and 0.29cm) -- cycle;
        
        % Back edges of equator band (dashed, thicker)
        \draw[dashed, bandcolor!70, thick] (-0.98,0.22) arc (180:360:0.98cm and 0.29cm);
        \draw[dashed, bandcolor!70, thick] (-0.98,-0.22) arc (180:360:0.98cm and 0.29cm);
        
        % Back meridian band (thicker)
        \fill[bandcolor!40, opacity=0.3] 
            (-0.22,-0.98) arc (270:90:0.29cm and 0.98cm) -- 
            (0.22,0.98) arc (90:270:0.29cm and 0.98cm) -- cycle;
        
        % Back edges of meridian band (dashed, thicker)
        \draw[dashed, bandcolor!70, thick] (-0.22,-0.98) arc (270:90:0.29cm and 0.98cm);
        \draw[dashed, bandcolor!70, thick] (0.22,-0.98) arc (270:90:0.29cm and 0.98cm);

% Front meridian (solid)
        \draw[gray!70, thick] (0,-1) arc (270:90:-0.3cm and 1cm);
        
        % Front meridian band (thicker)
        \fill[bandcolor!70, opacity=0.5] 
            (-0.22,-0.98) arc (270:90:-0.29cm and 0.98cm) -- 
            (0.22,0.98) arc (90:270:-0.29cm and 0.98cm) -- cycle;
        
        % Front edges of meridian band (solid, thicker)
        \draw[bandcolor!70!black, thick] (-0.22,-0.98) arc (270:90:-0.29cm and 0.98cm);
        \draw[bandcolor!70!black, thick] (0.22,-0.98) arc (270:90:-0.29cm and 0.98cm);
        
        % Front part of equator (solid, gray)
        \draw[thick, gray!70] (-1,0) arc (180:0:1cm and 0.3cm);
        
        % Front part of equator band (thicker)
        \fill[bandcolor!70, opacity=0.5] 
            (-0.98,0.22) arc (180:0:0.98cm and 0.29cm) -- 
            (0.98,-0.22) arc (0:180:0.98cm and 0.29cm) -- cycle;
        
        % Front edges of equator band (solid, thicker)
        \draw[bandcolor!70!black, thick] (-0.98,0.22) arc (180:0:0.98cm and 0.29cm);
        \draw[bandcolor!70!black, thick] (0.98,-0.22) arc (0:180:0.98cm and 0.29cm);

        \draw[very thick,red,rotate=15] (0.63,0.21) arc (0:360:0.288cm and 0.36cm);
        
        % Outer circle (limb of sphere)
        \draw[very thick] (0,0) circle (1cm);
        \node[below, align=center] at (0,1.1) {$\partial J^+(\tilde p)$ at\\ $t=t_{\mathrm{max}}$};
    \end{scope}
    
\end{tikzpicture}

%% file: Tikz4.tex
\begin{center}
\begin{tikzpicture}[scale=1.25]
    % Define colors

    % === FIRST GRAPH: 
    \begin{scope}[xshift=0cm, x={(1cm,0cm)}, y={(0.2cm,0.25cm)}, z={(0cm,1cm)}]
    % === STYLES ===
    \tikzset{
        axis/.style={thick, gray, -stealth},
        grid line/.style={thin, gray!40},
        cone edge/.style={thick},
        cone edge hidden/.style={thick, dashed, gray!60}
    }
    
    % === PARAMETERS ===
    % Define circle parameters
    \def\angleA{110}
    \def\angleB{220} % Changed to demonstrate a larger arc
    \def\angleC{130}
    \def\angleD{200}
    \def\axislen{3.5}
    \def\coneheight{2.5}
    \def\coneradius{1.1}
    \def\ellipseminor{1.1}  % Minor axis of ellipse (y-direction projection)
    
    % === COORDINATE SYSTEM AXES ===
    % x-axis
    \draw[axis] (-1,1,0) -- (-1+0.5*\axislen,1,0) node[right] {};
    % y-axis
    \draw[axis] (-1,1,0) -- (-1,1-0.8*\axislen,0) node[above right] {};
    % z-axis
    \draw[axis] (-1,1,0) -- (-1,1,0.8*\axislen) node[above] {};
    
    % === CONE COORDINATES ===
    \coordinate (apex) at (0,0,0);
    \coordinate (base) at (0,0,\coneheight);
    
    % Edge points on the base ellipse
    \coordinate (baseRight) at (\coneradius-0.02,0.1,\coneheight);
    \coordinate (baseFront) at (0,-\ellipseminor,\coneheight);
    \coordinate (baseLeft) at (-\coneradius+0.07,-0.2,\coneheight);
    \coordinate (baseBack) at (0,\ellipseminor,\coneheight);
    \coordinate (basemidpoint) at (0,0.81*\ellipseminor,0.81*\coneheight);
    
    % Points on circle using parametric coordinates (for marking start/end points)
    \coordinate (PointA) at ({sin(\angleA)*\coneradius},{cos(\angleA)*\ellipseminor},\coneheight);
    \coordinate (PointB) at ({sin(\angleB)*\coneradius},{cos(\angleB)*\ellipseminor},\coneheight);
     \coordinate (PointC) at ({sin(\angleC)*\coneradius},{cos(\angleC)*\ellipseminor},\coneheight);
    \coordinate (PointD) at ({sin(\angleD)*\coneradius},{cos(\angleD)*\ellipseminor},\coneheight);
     \coordinate (PointA2) at ({0.6*sin(\angleA)*\coneradius},{0.6*cos(\angleA)*\ellipseminor},{0.6*\coneheight});
    \coordinate (PointB2) at ({0.6*sin(\angleB)*\coneradius},{0.6*cos(\angleB)*\ellipseminor},{0.6*\coneheight});
     \coordinate (PointC2) at ({0.6*sin(\angleC)*\coneradius},{0.6*cos(\angleC)*\ellipseminor},{0.6*\coneheight});
    \coordinate (PointD2) at ({0.6*sin(\angleD)*\coneradius},{0.6*cos(\angleD)*\ellipseminor},{0.6*\coneheight});

  % Define the start point of the arc within this local coordinate system
        \coordinate (startArcPoint) at ({sin(\angleA)*\coneradius},{cos(\angleA)*\ellipseminor},\coneheight);
        
 % Draw the elliptical arc
        \draw[thick] (startArcPoint) 
            arc[start angle={90-\angleA}, end angle={90-\angleB}, x radius=\coneradius, y radius=\ellipseminor];

  % Define the start point of the arc within this local coordinate system
        \coordinate (startArcPoint) at ({0.6*sin(\angleA)*\coneradius},{0.6*cos(\angleA)*\ellipseminor},{0.6*\coneheight});
        
 % Draw the elliptical arc
        \draw[thick] (startArcPoint) 
            arc[start angle={90-\angleA}, end angle={90-\angleB}, x radius={0.6*\coneradius}, y radius={0.6*\ellipseminor}];
    
    % === VISIBLE CONE EDGES (solid) ===
    % Right edge
    \draw[cone edge, decorate, decoration={snake, amplitude=0.5mm, segment length=6.3mm}] (PointA2) -- (PointA);
    % Front edge
    \draw[cone edge,decorate, decoration={snake, amplitude=0.5mm, segment length=6.3mm}] (PointB2) -- (PointB);
    
    % Mark points
    %\fill[red] (PointC2) circle (3pt);
   % \fill[red] (PointD2) circle (3pt);

\begin{scope}[shift={(0,0,{0.92*\coneheight})}]
        % Define the start point of the arc within this local coordinate system
        \coordinate (label) at (0,0,0);
        
\node[above] (label){$\Omega\subset \partial J^+(p)$};
    \end{scope}

    \end{scope}
    
    % === SECOND GRAPH
    \begin{scope}[xshift=3cm, x={(1cm,0cm)}, y={(0.2cm,0.25cm)}, z={(0cm,1cm)}]
    % === STYLES ===
    \tikzset{
        axis/.style={thick, gray, -stealth},
        grid line/.style={thin, gray!40},
        cone edge/.style={thick},
        cone edge hidden/.style={thick, dashed, gray!60}
    }
    
    % === PARAMETERS ===
    % Define circle parameters
    \def\angleA{110}
    \def\angleB{220} % Changed to demonstrate a larger arc
    \def\angleC{130}
    \def\angleD{200}
    \def\axislen{3.5}
    \def\coneheight{2.5}
    \def\coneradius{1.1}
    \def\ellipseminor{1.1}  % Minor axis of ellipse (y-direction projection)
    
    % === COORDINATE SYSTEM AXES ===
    % x-axis
    \draw[axis] (-1,1,0) -- (-1+0.5*\axislen,1,0) node[right] {};
    % y-axis
    \draw[axis] (-1,1,0) -- (-1,1-0.8*\axislen,0) node[above right] {};
    % z-axis
    \draw[axis] (-1,1,0) -- (-1,1,0.8*\axislen) node[above] {};
    
    % === CONE COORDINATES ===
    \coordinate (apex) at (0,0,0);
    \coordinate (base) at (0,0,\coneheight);
    
    % Edge points on the base ellipse
    \coordinate (baseRight) at (\coneradius-0.02,0.1,\coneheight);
    \coordinate (baseFront) at (0,-\ellipseminor,\coneheight);
    \coordinate (baseLeft) at (-\coneradius+0.07,-0.2,\coneheight);
    \coordinate (baseBack) at (0,\ellipseminor,\coneheight);
    \coordinate (basemidpoint) at (0,0.81*\ellipseminor,0.81*\coneheight);
    
    % Points on circle using parametric coordinates (for marking start/end points)
    \coordinate (PointA) at ({sin(\angleA)*\coneradius},{cos(\angleA)*\ellipseminor},\coneheight);
    \coordinate (PointB) at ({sin(\angleB)*\coneradius},{cos(\angleB)*\ellipseminor},\coneheight);
     \coordinate (PointC) at ({sin(\angleC)*\coneradius},{cos(\angleC)*\ellipseminor},\coneheight);
    \coordinate (PointD) at ({sin(\angleD)*\coneradius},{cos(\angleD)*\ellipseminor},\coneheight);
     \coordinate (PointA2) at ({0.6*sin(\angleA)*\coneradius},{0.6*cos(\angleA)*\ellipseminor},{0.6*\coneheight});
    \coordinate (PointB2) at ({0.6*sin(\angleB)*\coneradius},{0.6*cos(\angleB)*\ellipseminor},{0.6*\coneheight});
     \coordinate (PointC2) at ({0.6*sin(\angleC)*\coneradius},{0.6*cos(\angleC)*\ellipseminor},{0.6*\coneheight});
    \coordinate (PointD2) at ({0.6*sin(\angleD)*\coneradius},{0.6*cos(\angleD)*\ellipseminor},{0.6*\coneheight});

  % Define the start point of the arc within this local coordinate system
        \coordinate (startArcPoint) at ({sin(\angleA)*\coneradius},{cos(\angleA)*\ellipseminor},\coneheight);
        
 % Draw the elliptical arc
        \draw[thick] (startArcPoint) 
            arc[start angle={90-\angleA}, end angle={90-\angleB}, x radius=\coneradius, y radius=\ellipseminor];

  % Define the start point of the arc within this local coordinate system
        \coordinate (startArcPoint) at ({0.6*sin(\angleA)*\coneradius},{0.6*cos(\angleA)*\ellipseminor},{0.6*\coneheight});
        
 % Draw the elliptical arc
        \draw[thick] (startArcPoint) 
            arc[start angle={90-\angleA}, end angle={90-\angleB}, x radius={0.6*\coneradius}, y radius={0.6*\ellipseminor}];
    
    % === VISIBLE CONE EDGES (solid) ===
    % Right edge
    \draw[cone edge, decorate, decoration={snake, amplitude=0.5mm, segment length=6.3mm}] (PointA2) -- (PointA);
    % Front edge
    \draw[cone edge,decorate, decoration={snake, amplitude=0.5mm, segment length=6.3mm}] (PointB2) -- (PointB);

    \draw[blue, thick](PointC2)--(PointC);
    \draw[blue, thick](PointD2)--(PointD);
    % Mark points
   \begin{scope}[shift={(0,0,{0.92*\coneheight})}]
        % Define the start point of the arc within this local coordinate system
        \coordinate (label) at (0,0,0);
        
\node[above] (label){$\Omega\subset \partial J^+(p)$};
    \end{scope}

        \begin{scope}[shift={(0,0,{0.53*\coneheight})}]
        % Define the start point of the arc within this local coordinate system
        \coordinate (label) at (0,0,0);
        
\node[below, align=center, scale =.75] (label){find two \\null geodesic \\segments};
    \end{scope}
      
    \end{scope}
    
    % === THIRD GRAPH
    \begin{scope}[xshift=6cm, x={(1cm,0cm)}, y={(0.2cm,0.25cm)}, z={(0cm,1cm)}]
    % === STYLES ===
    \tikzset{
        axis/.style={thick, gray, -stealth},
        grid line/.style={thin, gray!40},
        cone edge/.style={thick},
        cone edge hidden/.style={thick, dashed, gray!60}
    }
    
    % === PARAMETERS ===
    % Define circle parameters
    \def\angleA{110}
    \def\angleB{220} % Changed to demonstrate a larger arc
    \def\angleC{130}
    \def\angleD{200}
    \def\axislen{3.5}
    \def\coneheight{2.5}
    \def\coneradius{1.1}
    \def\ellipseminor{1.1}  % Minor axis of ellipse (y-direction projection)
    
    % === COORDINATE SYSTEM AXES ===
    % x-axis
    \draw[axis] (-1,1,0) -- (-1+0.5*\axislen,1,0) node[right] {};
    % y-axis
    \draw[axis] (-1,1,0) -- (-1,1-0.8*\axislen,0) node[above right] {};
    % z-axis
    \draw[axis] (-1,1,0) -- (-1,1,0.8*\axislen) node[above] {};
    
    % === CONE COORDINATES ===
    \coordinate (apex) at (0,0,0);
    \coordinate (base) at (0,0,\coneheight);
    
    % Edge points on the base ellipse
    \coordinate (baseRight) at (\coneradius-0.02,0.1,\coneheight);
    \coordinate (baseFront) at (0,-\ellipseminor,\coneheight);
    \coordinate (baseLeft) at (-\coneradius+0.07,-0.2,\coneheight);
    \coordinate (baseBack) at (0,\ellipseminor,\coneheight);
    \coordinate (basemidpoint) at (0,0.81*\ellipseminor,0.81*\coneheight);
    
    % Points on circle using parametric coordinates (for marking start/end points)
    \coordinate (PointA) at ({sin(\angleA)*\coneradius},{cos(\angleA)*\ellipseminor},\coneheight);
    \coordinate (PointB) at ({sin(\angleB)*\coneradius},{cos(\angleB)*\ellipseminor},\coneheight);
     \coordinate (PointC) at ({sin(\angleC)*\coneradius},{cos(\angleC)*\ellipseminor},\coneheight);
    \coordinate (PointD) at ({sin(\angleD)*\coneradius},{cos(\angleD)*\ellipseminor},\coneheight);
     \coordinate (PointA2) at ({0.6*sin(\angleA)*\coneradius},{0.6*cos(\angleA)*\ellipseminor},{0.6*\coneheight});
    \coordinate (PointB2) at ({0.6*sin(\angleB)*\coneradius},{0.6*cos(\angleB)*\ellipseminor},{0.6*\coneheight});
     \coordinate (PointC2) at ({0.6*sin(\angleC)*\coneradius},{0.6*cos(\angleC)*\ellipseminor},{0.6*\coneheight});
    \coordinate (PointD2) at ({0.6*sin(\angleD)*\coneradius},{0.6*cos(\angleD)*\ellipseminor},{0.6*\coneheight});
       \coordinate (PointC3) at ({-0.1*sin(\angleC)*\coneradius},{-0.1*cos(\angleC)*\ellipseminor},{-0.1*\coneheight});
    \coordinate (PointD3) at ({-0.1*sin(\angleD)*\coneradius},{-0.1*cos(\angleD)*\ellipseminor},{-0.1*\coneheight});
  
  % Define the start point of the arc within this local coordinate system
        \coordinate (startArcPoint) at ({sin(\angleA)*\coneradius},{cos(\angleA)*\ellipseminor},\coneheight);
        
 % Draw the elliptical arc
        \draw[thick] (startArcPoint) 
            arc[start angle={90-\angleA}, end angle={90-\angleB}, x radius=\coneradius, y radius=\ellipseminor];

  % Define the start point of the arc within this local coordinate system
        \coordinate (startArcPoint) at ({0.6*sin(\angleA)*\coneradius},{0.6*cos(\angleA)*\ellipseminor},{0.6*\coneheight});
        
 % Draw the elliptical arc
        \draw[thick] (startArcPoint) 
            arc[start angle={90-\angleA}, end angle={90-\angleB}, x radius={0.6*\coneradius}, y radius={0.6*\ellipseminor}];
    
    % === VISIBLE CONE EDGES (solid) ===
    % Right edge
    \draw[cone edge, decorate, decoration={snake, amplitude=0.5mm, segment length=6.3mm}] (PointA2) -- (PointA);
    % Front edge
    \draw[cone edge,decorate, decoration={snake, amplitude=0.5mm, segment length=6.3mm}] (PointB2) -- (PointB);

 \draw[blue, thick](PointC3)--(PointC);
    \draw[blue, thick](PointD3)--(PointD);

    % Mark points
    \fill[red] (apex) circle (1.5pt);
   
      \node[right](apex) {$p$};

      \begin{scope}[shift={(0,0,{0.92*\coneheight})}]
        % Define the start point of the arc within this local coordinate system
        \coordinate (label) at (0,0,0);
        
\node[above] (label){$\Omega\subset \partial J^+(p)$};
    \end{scope}

    \begin{scope}[shift={(0,0,{0.3*\coneheight})}]
        % Define the start point of the arc within this local coordinate system
        \coordinate (label) at (0,0,0);
        
\node[right, scale = .75] (label){\hphantom{dd}past extension };
    \end{scope}
  
    \end{scope}

\end{tikzpicture}
\end{center}